\documentclass[reprint,amsfonts,amssymb,amsmath,showkeys,nofootinbib,aps,pra, superscriptaddress,twocolumn,longbibliography]{revtex4-2}

\usepackage[english]{babel}
\usepackage[utf8]{inputenc}
\usepackage[T1]{fontenc}
\usepackage{amsthm}
\usepackage{mathtools}
\usepackage{xcolor}
\usepackage{graphicx}
\usepackage{adjustbox}
\usepackage{tikz}

\usepackage{dsfont}

\usepackage{longtable}
\usepackage{multirow}

\usepackage{bbm}

\newcommand\mbb[1]{\mathbb{#1}}

\def\HC{\mathcal{H}}

\def\LC{\mathcal{L}}

\def\ad{^{\dagger}}

\usepackage[colorlinks=true,citecolor=blue,linkcolor=magenta]{hyperref}

\usepackage{xspace}

\newcommand{\BC}{\mathcal{B}}

\newcommand{\GC}{\mathcal{G}}

 \newcommand{\Tr}{{\rm Tr}}

\renewcommand{\geq}{\geqslant}

\renewcommand{\vec}[1]{\boldsymbol{#1}}  

\newcommand*{\id}{\openone}

\newcommand{\bs}{\textsf{BS}}

\newcommand{\lm}{\lambda }

\newcommand{\thv}{\vec{\theta}}

\newcommand{\losalamos}{Theoretical Division, Los Alamos National Laboratory, Los Alamos, New Mexico 87545, USA}
\newcommand{\courant}{Courant Institute of Mathematical Sciences, New York University, New York, New York 10012, USA}

\def\be{\begin{equation}}
\def\ee{\end{equation}}
\def\bs{\begin{split}}
\def\e{\end{split}}
\def\ba{\begin{eqnarray}}
\def\bea{\begin{eqnarray}}

\def\tea{\end{eqnarray}}
\def\ea{\end{eqnarray}}
\def\eea{\end{eqnarray}}

\def\comm{\mf{comm}}

\def\tn{^{\otimes n}}
\def\tn{^{\otimes n}}

\def\comm{\mf{comm}}

\newcommand\mf[1]{\mathfrak{#1}}

\newcommand\spn{\mathrm{span}}
\def\Lie{{\rm Lie}}

\newtheorem{theorem}{Theorem}

\newtheorem{lemma}{Lemma}
\newtheorem{corollary}{Corollary}

\newtheorem{proposition}{Proposition}

\newtheorem{definition}{Definition}

\def\be{\begin{equation}}
\def\te{\end{equation}}
\def\ee{\end{equation}}
\def\ba{\begin{eqnarray}}
\def\bea{\begin{eqnarray}}

\def\tea{\end{eqnarray}}
\def\ea{\end{eqnarray}}
\def\eea{\end{eqnarray}}

\begin{document}

\title{On the universality of $S_n$-equivariant $k$-body gates}

\author{Sujay Kazi}
\thanks{sujay.kazi@gmail.com}
\affiliation{\courant}

\author{Mart\'{i}n Larocca}
\affiliation{\losalamos}
\affiliation{Center for Nonlinear Studies, Los Alamos National Laboratory, Los Alamos, New Mexico 87545, USA}

\author{M. Cerezo}
\thanks{cerezo@lanl.gov}
\affiliation{Information Sciences, Los Alamos National Laboratory, Los Alamos, NM 87545, USA}

\begin{abstract}
 The importance of symmetries has recently been recognized in  quantum machine learning from the simple motto: if a task exhibits a symmetry (given by a group $\mathfrak{G}$), the learning model should respect said symmetry. This can be instantiated via $\mathfrak{G}$-equivariant Quantum Neural Networks (QNNs), i.e., parametrized quantum circuits whose gates are generated by operators commuting with a given representation of $\mathfrak{G}$. In practice, however, there might be additional restrictions to the types of gates one can use, such as being able to act on at most $k$ qubits. In this work we study how the interplay between symmetry and $k$-bodyness in the QNN generators affect its expressiveness for the special case of $\mathfrak{G}=S_n$, the symmetric group. Our results show that if the QNN is generated by one- and two-body $S_n$-equivariant gates, the QNN is semi-universal but not universal. That is, the QNN can generate any arbitrary special unitary matrix in the invariant subspaces, but has no control over the relative phases between them.  Then, we show that in order to reach universality one needs to include $n$-body generators  (if $n$ is even) or $(n-1)$-body generators (if $n$ is odd). As such, our results brings us a step closer to better understanding the capabilities and limitations of equivariant QNNs.
\end{abstract}

\maketitle

\let\clearpage\relax

\section{Introduction}

Quantum Machine Learning (QML) aims to use quantum computers to classify, cluster, and make predictions from  classical data encoded in quantum states, or from quantum data produced by some quantum mechanical process~\cite{biamonte2017quantum,schuld2015introduction}. Despite the tremendous attention QML has received, its true potential and limitations are still unclear~\cite{cerezo2022challenges}.  Recent results have shown that QML models lacking inductive biases (i.e., models that do not contain information about the specific problem being tackled) can encounter serious trainability and generalization issues~\cite{cerezo2020cost,holmes2021connecting,kubler2021inductive,larocca2021diagnosing}. This can be understood as a form of ``no-free-luch''  whereby high-expressiveness, multi-purpose algorithms will have overall poor performance.

Numerous endeavors have been undertaken to create learning models that are tailored specifically to a given task. Among these, Geometric Quantum Machine Learning (GQML) has emerged as one of the most promising approaches~\cite{larocca2022group,meyer2022exploiting,skolik2022equivariant,nguyen2022atheory,schatzki2022theoretical,ragone2022representation,sauvage2022building,anschuetz2022efficient}. The fundamental idea behind GQML is to leverage the symmetries present in the task to develop sharp inductive biases for the learning models. For instance, when dealing with the task of classifying $n$-qubit pure states from $n$-qubit mixed states, GQML suggests using models whose outputs remain invariant under arbitrary unitaries applied to the qubits~\cite{larocca2022group}. This is because these unitaries cannot change the purity and are therefore symmetries of the task.

The GQML program consists of several steps. First, one needs to identify the group of transformations preserving some important property of the data (e.g., a symmetry that preserves the labels in supervised learning). In the case of purity, the group is $\mf{G}=\mathbb{U}(2^n)$. The next step is to build a model that generates labels that remain invariant under $\mf{G}$. Several recent works provide templates for developing such $\mf{G}$-invariant models, which can be found in Refs.~\cite{meyer2022exploiting,nguyen2022atheory}. For models based on Quantum Neural Networks (QNNs), one needs to leverage the concept of \textit{equivariance}~\cite{bronstein2021geometric,nguyen2022atheory}.

An equivariant QNN is a parametrized quantum circuit ansatz that commutes, for all values of parameters, with the representation of $\mf{G}$.  Figure~\ref{fig:Overview} illustrates this constraint on the set of unitaries that the QNN can express (i.e., to the subspace of $\mf{G}$-\textit{symmetric} operators). Given the practical limitations on implementing many-body gates, it is reasonable to ask whether having access to elementary up-to-$k$-local equivariant gates suffices for the QNN to express \textit{any} possible equivariant unitary. This question was first studied in Ref.~\cite{marvian2022restrictions}, where the authors found that, as opposed to the well-known result which  states that any unitary can be decomposed in one- and two-qubit gates, certain equivariant unitaries cannot be decomposed in equivariant $k$-local gates. These result were then extended in Ref.~\cite{marvian2022rotationally,marvian2021qudit,marvian2023non} to rotationally-invariant circuits, tensor product representations of groups acting on qudits, and tensor product representations of Abelian symmetries.

In this work, we will focus on $\mf{G}=S_n$, the symmetric group of permutations, with its action on $n$ qubits. This group is of special interest as it is the relevant symmetry group for a wide range of learning tasks related to problems defined on sets, graphs and grids, molecular systems, multipartite entanglement, and distributed quantum sensing~\cite{maron2020learning,maron2018invariant, keriven2019universal, maron2019provably,verdon2019quantumgraph,cong2019quantum,beckey2021computable,guo2020distributed,huerta2022inference,anschuetz2022efficient}. Our results can be summarized as follows. First, we show that that $S_n$-equivariant QNNs with $2$-body gates are enough to reach semi-universality (and hence subspace controllability), but not universality. This means that the circuit can generate \textit{any} arbitrary special unitary matrix in each of the subspaces but that one has no control over all their relative phases. We demonstrate that this small, albeit important, difference follows from the central projection theorem of Ref.~\cite{zimboras2015symmetry}. Next, we prove that if one wishes to attain subspace universality, then the QNN must contain $n$-body $S_n$-equivariant gates (if $n$ is even) or $(n-1)$-body gates (if $n$ is odd). Hence, our theorems impose restrictions on the set of unitaries that can be generated when combining symmetry and few-bodyness. We further argue that our results are important for the field of quantum optimal control~\cite{d2007introduction} as they correct a previous result in the literature~\cite{albertini2018controllability}.

\section{Results}

\subsection{Background: Dynamical Lie algebra and dynamical Lie group for QNNs}

Parametrized quantum circuits, or more generally, QNNs, are one of the most widely used computational paradigms to process information in near-term quantum computers due to their high versatility~\cite{cerezo2020variationalreview,cerezo2022challenges}. In what follows, we will consider QNNs acting on systems of $n$ qubits (with associated $d=2^n$-dimensional Hilbert space $\HC$) of the form
\begin{equation}\label{eq:qnn}
U(\thv)=\prod_{m=1}^M e^{-i \theta_m H_m}\,,
\end{equation}
where $H_m$ are Hermitian operators taken from a given \textit{set of generators} $\GC$, and $\thv=(\theta_1,\ldots\theta_M)\in\mathbb{R}^M$ are trainable parameters. Here we note that the operators in $\GC$ constitute a veritable fingerprint for the QNN, as they allow us to differentiate a given architecture from another.

\begin{figure}[t]
    \centering
\includegraphics[width=.94\columnwidth]{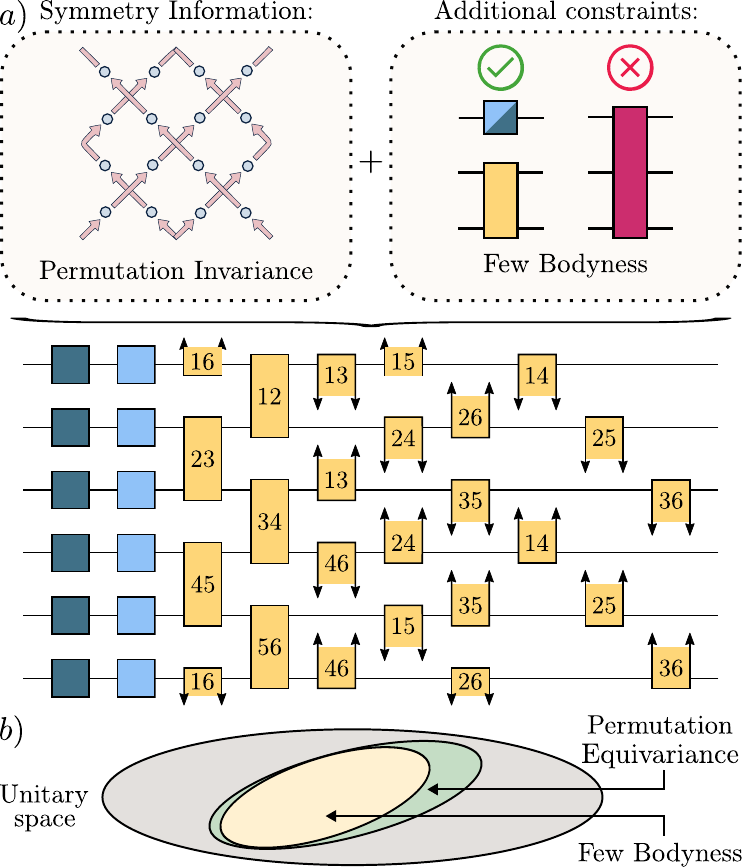}
    \caption{\textbf{Summary of our main results.} a) Recent results in the field of GQML indicate that if the problem has a given relevant symmetry $\mf{G}$, then the gates in the QNN should be $\mf{G}$-equivariant~\cite{nguyen2022atheory}. In this work we consider the case where $\mf{G}=S_n$, the permutation group, and study how the expressiveness of the QNN changes when one imposes the additional constraint of few-bodyness on the QNN's gates. In the circuit, gates with the same color share a common parameter. b) The effect of restricting the set of available elementary gates in the QNN is to restrict its expressiveness, i.e., how much of the unitary space it can cover by varying its parameters. Imposing $\mf{G}$-equivariance can appropriately reduce the QNN's expressiveness to a region of unitaries respecting the task's symmetry. Imposing additional restrictions, such as few-bodyness, further restricts its expressiveness.
    \label{fig:Overview}}
\end{figure}

In this work we consider the case where $\GC$ contains only up-to-$k$-body gates, i.e., gates acting non-trivially on at most $k$ qubits. Such restriction is usually physically motivated and arises when working with gates that are  native to some specific hardware~\cite{kandala2017hardware}. At this point, we find it convenient to recall that in the absence of symmetries, $2$-local gates are universal~\cite{divincenzo1995two,lloyd1995almost}. That is, they suffice to generate any $d$-dimensional unitary. However, the same is generally not true when the operators in $\GC$ are local, but also chosen to respect a certain symmetry group $\mf{G}$~\cite{marvian2022restrictions,marvian2023non}. Hence, given these constraints, it is critical to quantify the QNN's expressiveness, i.e., the breadth of unitaries that $U(\thv)$ can generate when varying the parameters $\thv$ (see Fig.~\ref{fig:Overview}(b)). By inspecting Eq.~\eqref{eq:qnn} one can see that two main factors come into play: the number of parameters $M$, and the set of generators $\GC$. While several measures of expressiveness exist~\cite{sim2019expressibility,holmes2021connecting,zimboras2015symmetry}, here we will focus on the so-called Dynamical Lie Algebra (DLA)~\cite{zeier2011symmetry}, which captures the potential expressiveness of the QNN (under arbitrary depth, or number of parameters).

\begin{figure*}[t]
    \centering
\includegraphics[width=.94\linewidth]{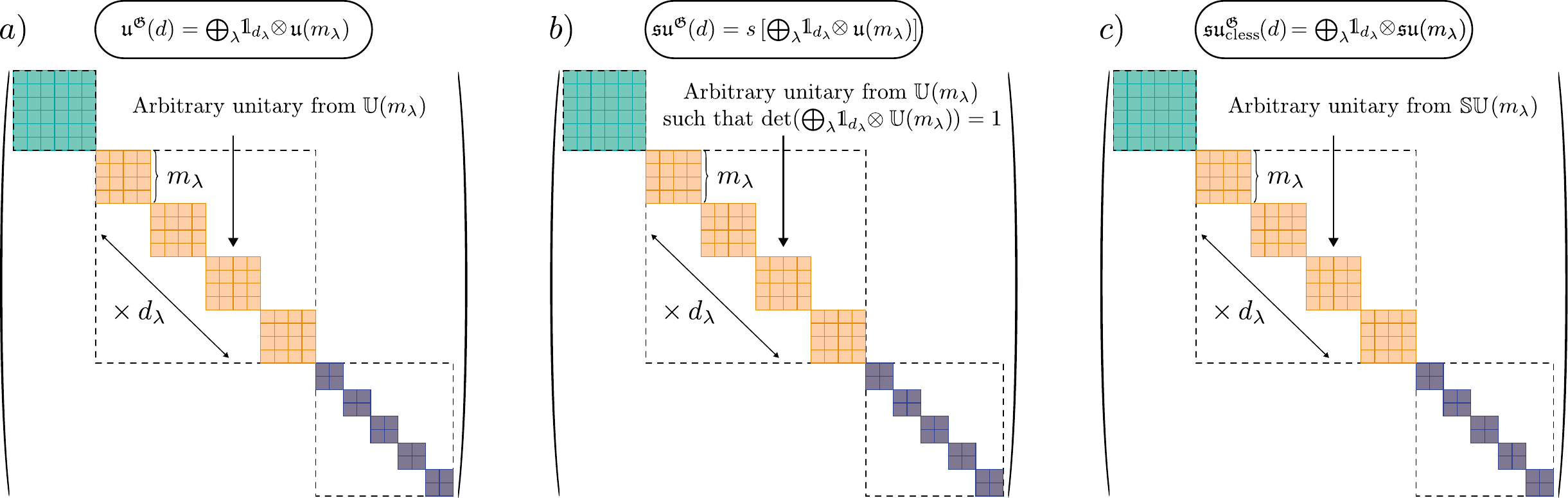}
    \caption{\textbf{Important Lie algebras and the irrep structure of the elements in the associated Lie groups}. In the main text we have defined three important subalgebras. a) The first is the maximal $\mf{G}$-symmetric subalgebra  $\mf{u}^{\mf{G}}(d)$. As schematically  shown, the associated Lie group contains the set of matrices where an arbitrary unitary can be prepared in each isotypic component. b) Next, we define the special subalgebra of $\mf{u}^{\mf{G}}(d)$, which we denote as $\mf{su}^{\mf{G}}(d)$. Now the associated Lie group contains all the matrices where an arbitrary unitary can be prepared in each isotypic component, under the additional restriction that the determinant of the overall matrix must be equal to one. c) Finally, we define the centerless subalgebra of $\mf{u}^{\mf{G}}(d)$, which we denote as $\mf{su}^{\mf{G}}_{\rm{cless}}(d)$. In this case, the Lie group contains all matrices where an arbitrary special unitary can be prepared in each isotypic component. Note that the difference between $\mf{su}^{\mf{G}}_{\rm{cless}}(d)$ and $\mf{u}^{\mf{G}}(d)$ is that the former does not allow us to control the relative phases between the isotypic components. }
    \label{fig:algebras}
\end{figure*}

Given a set of Hermitian generators $\GC$, the DLA is the subspace of operator space $\mf{u}(d)$ spanned by the repeated nested commutators of the elements in $\GC$. That is,
\begin{equation}
\mf{g} = \langle i \GC\rangle_{\Lie}\,,
\end{equation}
where $\langle \,\cdot\,\rangle_{\Lie}$ denotes the Lie closure.
Notably, the DLA fully determines the ultimate expressiveness of the QNN, as we have
\begin{equation}
U(\thv)\in \mathbb{G}=e^{\mf{g}}\subseteq \mathbb{U}(d)\,.
\end{equation} 
Here, $\mathbb{G}$ denotes the dynamical Lie group, and is composed of any possible unitary generated by a unitary $U(\thv)$ as in Eq.~\eqref{eq:qnn} for any possible choice of $\thv$ and $M$. When $\mathbb{G}=\mathbb{SU}(d)$ or $\mathbb{G}=\mathbb{U}(d)$, the QNN is said to be \textit{controllable}, or \textit{universal}, as any unitary can be generated from $U(\thv)$ (up to a global phase)~\cite{d2007introduction,schirmer2003controllability}. In what follows, we will study how constraints in $\GC$ (such as equivariance or few-bodyness) alter the QNN's ability to become universal. For this purpose, we recall the mathematical definition of $\mf{G}$-\textit{equivariance}~\cite{nguyen2022atheory,schatzki2022theoretical,ragone2022representation,meyer2022exploiting}. Given a  symmetry group $\mf{G}$ (which we henceforth assume to be compact), we have the following definition.
\begin{definition}[Equivariant operators]\label{def:equiv}
    An operator $H$ is said to be $\mf{G}$-\textit{equivariant} if it commutes with a representation of a given compact group $\mf{G}$. That is, if
\begin{equation}\label{eq:equiv}
[H,R(g)]=0\,, \quad g\in\mf{G}\,,
\end{equation}
where $R$ is a unitary representation of $\mf{G}$ on the $d$-dimensional Hilbert space $\HC$. 
\end{definition}
Note that if all the generators in $\GC$ are equivariant, the QNN can be readily shown to be equivariant itself~\cite{nguyen2022atheory,meyer2022exploiting}. We also highlight the fact that Definition~\ref{def:equiv} implies 
\begin{equation}
H\in\comm(\mf{G})\,,
\end{equation}
where $\comm({\mf{G}})$ denotes the \textit{commutant algebra} of the (representation $R$ of the) group $\mf{G}$, i.e., the associative matrix algebra of linear operators that commute with every element in $\mf{G}$:
\begin{equation}
    \comm({\mf{G}})=\{A\in\BC(\HC)\,\,|\,\,[A,R(g)]=0\,,\,\,\forall g\in\mf{G}\}\,.
\end{equation}
Here, $\BC(\HC)$ denotes  the space of bounded linear operators in $\HC$.

Next, it is fundamental to recall that the representation $R$ admits an isotypic decomposition
\begin{equation}
R(g\in\mf{G})\cong\bigoplus_{\lambda=1}^L r_\lambda(g)\otimes \id_{m_\lambda}\,,    
\end{equation}
where $r_\lambda$ is a $d_\lambda=\dim(r_\lambda)$-dimensional irreducible representation (irrep) of $\mf{G}$, and $\id_{m_\lambda}$ is an identity of dimension $m_\lambda$. Hence, $m_\lambda$ denotes the multiplicity associated to each irrep. In this case, the Hilbert space $\HC$ can be expressed as
\begin{equation}
    \HC=\bigoplus_{\lambda,\nu} \HC_{\lambda}^\nu\,,
\end{equation}
where $\HC_{\lambda}^\nu$ is a $m_\lambda$-dimensional space, and $\nu=1,\cdots,d_\lambda$.  Moreover, using $\mbb{P}_{\lambda}^\nu$ to denote the projector onto the subspace $\HC_{\lambda}^\nu$, we can focus on the part of the DLA that acts non-trivially on each $\HC_{\lambda}^\nu$, given by
\begin{equation}
   \mf{g}_{\lambda}^\nu=\{\mbb{P}_{\lambda}^\nu\,iH \,,\,\,iH\in\mf{g}\}\,.
\end{equation}
Note that by definition $\mf{g}_{\lambda}^\nu=\mf{g}_{\lambda}^{\nu'}$, which allows us to simply use the notation $\mf{g}_{\lambda}:=\mf{g}_{\lambda}^\nu$.

The previous motivates us to define three important Lie subalgebras of $\mf{u}(d)$. First, we define the \textit{maximal} $\mf{G}$-symmetric subalgebra $\mf{u}^{\mf{G}}(d)=\comm(\mf{G})\cap \mf{u}(d)$ as 
\begin{align}\label{eq:maximal-g-sym-u}
\mf{u}^{\mf{G}}(d)&=Q\left(\bigoplus_{\lambda=1}^L\mf{u}(m_\lambda)\right)=\bigoplus_{\lambda=1}^L\id_{d_\lambda}\otimes \mf{u}(m_\lambda)\,,
\end{align}
where $Q$ is a representation defined by the right-hand-side of Eq.~\eqref{eq:maximal-g-sym-u}. We then define the  \textit{maximal special} subalgebra $\mf{su}^{\mf{G}}(d)=\comm(\mf{G})\cap \mf{su}(d)$, given by
\begin{align}\label{eq:maximal-g-sym-su}
\mf{su}^{\mf{G}}(d)&=s\left[\bigoplus_{\lambda=1}^L\id_{d_\lambda}\otimes \mf{u}(m_\lambda)\right]\,,
\end{align}
where $s[\cdot]$ denotes keeping the operators with vanishing trace. 
Finally, we also define the \textit{maximal centerless} $\mf{G}$-symmetric subalgebra
\small
\begin{align}\label{eq:maximal-g-sym}
\mf{su}^{\mf{G}}_{\rm{cless}}(d)&=Q\left(\bigoplus_{\lambda=1}^L\mf{su}(m_\lambda)\right)=\bigoplus_{\lambda=1}^L\id_{d_\lambda}\otimes \mf{su}(m_\lambda)\,.
\end{align}
\normalsize
In Fig.~\ref{fig:algebras} we show how the previous three algebras lead to different unitaries in their associated Lie groups. 

With the previous algebras in mind, we can introduce three key definitions that will allow us to study controllability and degrees of universality when there are symmetries in play.

\begin{definition}\label{def:subspace-control-new}
The QNN, or its associated DLA, is said to be:
\begin{itemize}
\item \textit{Subspace controllable}, if for all $\lambda$,
\begin{equation}
     \mf{su}(m_\lambda)\subseteq \mf{g}_\lambda\subseteq\mf{u}(m_\lambda)\,.
\end{equation}
\item \textit{Semi-universal}, if 
\begin{align}\label{eq:maximal-g-sym-su-min}
\mf{su}^{\mf{G}}_{\rm{cless}}(d)\subseteq\mf{g}\subseteq\mf{u}^{\mf{G}}(d)\,.
\end{align}
\item \textit{Universal}, or \textit{completely controllable}, if
\begin{equation}
     \mf{su}^{\mf{G}}(d)\subseteq \mf{g} \subseteq\mf{u}^{\mf{G}}(d)\,.
\end{equation}
\end{itemize}
\end{definition}

It is clear that complete controllability~\cite{d2007introduction} implies semi-universality~\cite{marvian2023non}, which in turn implies subspace controllability~\cite{wang2016subspace}. The key difference between the unitaries arising from a DLA which is semi-universal but \textit{not} universal is that one cannot control the \textit{relative} phases between the unitaries acting within each isotypic component (see Fig.~\ref{fig:algebras}). On the other hand, a QNN that is subspace controllable can generate arbitrary unitaries in each subspace, but not necessarily independently as in the semi-universal case. For instance, if $\HC=\HC_{\lambda_1}\otimes \HC_{\lambda_2}$ with $\dim(\HC_{\lambda_1})=\dim(\HC_{\lambda_2})=d/2$, then a QNN producing unitaries of the form $V\oplus V$ with $V\in e^{\mf{su}(d/2)}$  is subspace controllable but not semi-universal. As we will see below, distinguishing between these three cases will be extremely important to better understand the expressive capacity of  QNNs with equivariance and few-bodyness constraints.

To finish this section we recall a fundamental concept: Given a Lie algebra $\mf{h}\subseteq\mf{u}(d)$, its center $\mf{z}(\mf{h})$ is composed of all the elements in $\mf{h}$ that commute with every element in $\mf{h}$, i.e.,
\begin{equation}
    \mf{z}(\mf{h})=\{iH\in\mf{h}\,\,|\,\,[H,H']=0\,,\,\,\forall iH'\in\mf{h}\}\,.
\end{equation}
In particular, $\mf{z}(\mf{u}^{\mf{G}}(d))=Q\left( \bigoplus_{\lm=1}^L \mf{u}(1) \right)$ and it follows that every element $iE\in \mf{z}(\mf{u}^{\mf{G}}(d))$ has the form (see Ref.~\cite{sagan2001symmetric}, Theorem 1.7.8)  
\begin{equation}\label{eq:isotypic-center}
  iE  =\bigoplus_{\lambda=1}^L ie_\lambda \id_{d_\lambda}\otimes  \id_{m_\lambda}\,,
\end{equation}
where $e_\lambda \in\mathbb{R}$.  Moreover, it also follows that
\begin{equation}
    \dim\left(\mf{z}(\mf{u}^{\mf{G}}(d))\right)=L\,,
\end{equation}
i.e., there are as many elements in the center as irreps in Eq.~\eqref{eq:maximal-g-sym-u}. Finally, let us note that $\mf{su}^{\mf{G}}_{\rm{cless}}(d)$ is the maximal $\mf{G}$-symmetric centerless subalgebra of $\mf{u}^{\mf{G}}(d)$. That is,
\begin{equation}\label{eq:centerless-center}
\mf{u}^{\mf{G}}(d)=
\mf{su}^{\mf{G}}_{\rm{cless}}(d)\cup \mf{z}(\mf{u}^{\mf{G}}(d))\,.
\end{equation}

\subsection{Main results}

In what follows we will consider $\mf{G}$ to be the symmetric group $S_n$ and  $R$ the qubit-permuting representation of $S_n$ on $n$ qubits,
\begin{equation}
R(\pi\in S_n)\bigotimes_{i=1}^n |\psi_i\rangle = \bigotimes_{i=1}^n |\psi_{\pi^{-1}(i)}\rangle\,.
\end{equation}
Given a set of equivariant $k$-local generators, it is natural to ask:  \textit{can we achieve subspace controllability, or even (semi-)universality?} 
In this section we characterize the expressiveness of permutation-equivariant few-bodyness-constrained ansatze.

First, let us investigate the case of $1$-body symmetric operators. For instance, consider the (single element) set
\begin{equation}
    \GC_1=\left\{\sum_{j=1}^nX_j\right\}\,,
\end{equation}
where we use $X_j,Y_j,Z_j$ to denote the corresponding Pauli operator acting on the $j^\text{th}$ qubit. By inspection, $\sum_{j=1}^nX_j$ is $S_n$-equivariant, as it is invariant under any qubit permutation. Moreover, we can also see that the  corresponding DLA $\mf{g}_1$ is a representation of $\mf{u}(1)$. Next, consider the set
\begin{equation}
    \GC_1'=\left\{\sum_{j=1}^nX_j,\sum_{j=1}^nY_j\right\}\,,
\end{equation}
for which the DLA $\mf{g}_1'$ is a representation of $\mf{su}(2)$. These two cases show that, unsurprisingly, one cannot obtain subspace controllability (let alone semi-universality or universality) using $1$-body symmetric operators.

Let us now consider up-to-$2$-body $S_n$-equivariant generators. In particular, consider the set
\begin{equation}\label{eq:genertors-2-body}
    \GC_2=\left\{\sum_{j=1}^nX_j,\sum_{j=1}^nY_j,\sum_{j_1<j_2}^nZ_{j_1}Z_{j_2}\right\}\,.
\end{equation}
The expressiveness of $\GC_2$ is captured by the following theorem. 
\begin{theorem}\label{theo:subspace-not-universal}
Consider the set $\GC_2$ of $S_n$-equivariant generators in Eq.~\eqref{eq:genertors-2-body}. The associated  DLA is 
\begin{align}
\mf{g}_2&=Q\left(\bigoplus_{\lambda=1}^L\mf{su}(m_\lambda)\oplus\mf{u}(1)\right)\\
&=\mf{su}^{S_n}_{\rm{cless}}(d) \boxplus Q\left(\mf{u}(1)\right) \,,
\end{align}
where $\boxplus$ denotes the Minkowski sum\footnote{Given two sets $A$ and $B$, their Minkowski sum is defined as $A\boxplus B = \{a+b\,|\,a\in A, b\in B\}$.} and where $Q(\mf{u}(1))\subset \mf{z}(\mf{u}^{S_n}(d))$.

\end{theorem}

Theorem~\ref{theo:subspace-not-universal} has several important implications. First, it shows  that adding a single $S_n$-equivariant $2$-body operator to $\GC_1'$ significantly increases its expressiveness, as it leads to a DLA $\mf{g}_2$ that contains $\mf{su}^{S_n}_{\rm{cless}}(d)$ plus the span of an element in the center $\mf{z}(\mf{u}^{S_n}(d))$.

The second implication of Theorem~\ref{theo:subspace-not-universal} is that the system is semi-universal (and thus subspace controllable) according to Definition~\ref{def:subspace-control-new}.  However, we note that the previous result also implies that  $\mf{g}_2$  fails to be universal (see  Definition~\ref{def:subspace-control-new}). This is in contrast to the claims in Ref.~\cite{albertini2018controllability} which state that $\mf{g}_2$ is universal. To clarify this discrepancy, we show in the Supplemental Material that the proof of universality in Ref.~\cite{albertini2018controllability} contains an error.

While we leave the full proof of Theorem~\ref{theo:subspace-not-universal} for the Supplemental Material, here we will explain the intuition and main steps behind it. In particular, we will constructively show that $\mf{g}_2$ contains the centerless subalgebra $\mf{su}^{S_n}_{\rm{cless}}(d)$ plus the span of a single element in the center $\mf{z}\left(\mf{u}^{S_n}(d)\right)$. The main tool behind deriving the latter result is the central projections condition of~\cite{zimboras2015symmetry}.

Let us introduce some useful notation. First, we will define the $S_n$-symmetrized Pauli strings.

\begin{definition}[Symmetrized Pauli strings]
Define $P_{(k_x,k_y,k_z)}$ to be the operator corresponding to the sum of all distinct Pauli strings that have $k_x$ $X$ symbols, $k_y$ $Y$ symbols, and $k_z$ $Z$ symbols. 
\end{definition}

With this definition we can rewrite the set of generators in Eq.~\eqref{eq:genertors-2-body} as $\GC_2= \{P_{(1,0,0)},P_{(0,1,0)},P_{(0,0,2)}\}$. Clearly, the set of all $P_{(k_x,k_y,k_z)}$ spans the maximal $S_n$-equivariant operator algebra. That is,
\small
\begin{align}\label{eq:maximal-algebra-P}
\mf{u}^{S_n}(d) = \spn_{\mbb{R}}\{iP_{(k_x,k_y,k_z)} \,|\, k_x+k_y+k_z\le n\}\,.
\end{align}
\normalsize
Moreover, if we leave outside $P_{(0,0,0)}=\id$, we get the \textit{special} maximal subalgebra
\small
\begin{equation}
\mf{su}^{S_n}(d) = \spn_{\mbb{R}}\{iP_{(k_x,k_y,k_z)} \,|\, 0< k_x+k_y+k_z\le n\}.
\end{equation}
\normalsize

Next, let us define the following set of operators:

\begin{definition}
For each integer $0\le \mu\le\lfloor n/2\rfloor$, define:
\begin{equation}\label{eq:Ck}
    C_\mu = \sum_{a+b+c=\mu}\frac{(2a)!(2b)!(2c)!}{a!b!c!}P_{(2a,2b,2c)}\,.
\end{equation}
\end{definition}
Each $C_\mu$ consists exclusively of $P_{(k_x,k_y,k_z)}$ terms where $k_x, k_y, k_z$ are all even and add to $2\mu$. Thus, $C_\mu$ is a $2\mu$-body operator. For example, consider the case
\begin{align}
    C_2 =& 12(P_{(4,0,0)} + P_{(0,4,0)} + P_{(0,0,4)})\nonumber \\
    &+ 4(P_{(2,2,0)} + P_{(2,0,2)} + P_{(0,2,2)})\nonumber\,.
\end{align}
The significance of the $C_\mu$ operators is explained by the following proposition, proved in the Supplemental Material. (For additional properties of  the $C_\mu$ operators, we also refer the reader to  Ref.~\cite{marvian2022rotationally}, where they are referred to as $C_l$ with $l = 2\mu$.)

\begin{proposition}[Center of the $S_n$-equivariant Lie algebra]
\label{lem:Ck-span-Sn-center-main}
The center of the $S_n$-equivariant Lie algebra $\mf{u}^{S_n}(d)$ is spanned by the $C_\mu$ operators:
\begin{equation}
    \mf{z}(\mf{u}^{S_n}(d)) = \spn_{\mbb{R}}\Big\{i C_\mu \,|\, 0\le \mu\le \Big\lfloor\frac{n}{2}\Big\rfloor \Big\}.
\end{equation}
\end{proposition}

We can see from Eqs.~\eqref{eq:maximal-algebra-P} and~\eqref{eq:Ck} that the maximal $S_n$-equivariant centerless subalgebra $\mf{su}^{S_n}_{\rm{cless}}(d)$ is given by all linear combinations of the form
\begin{equation}
    i\sum_{0\le k_x+k_y+k_z\le n}c_{(k_x,k_y,k_z)}P_{(k_x,k_y,k_z)}\,,
\end{equation}
 with $c_{(k_x,k_y,k_z)}\in\mathbb{R}$, such that 
 \begin{equation}
     \sum_{a+b+c=\mu}\frac{c_{(2a,2b,2c)}}{a!b!c!} = 0 \,,
 \end{equation}
for each integer $0\le \mu\le\lfloor n/2\rfloor$.

As previously mentioned (and as seen in Eq.~\eqref{eq:centerless-center}), for a DLA to be universal it needs to generate all the elements in $\mf{su}^{S_n}_{\rm{cless}}(d)$ plus all the traceless elements in the center $\mf{z}(\mf{u}^{S_n}(d))$. Notably, one can prove the following necessary condition for a DLA to be able to generate elements in the center:

\begin{theorem}[Central Projections. Result 1 in~\cite{zimboras2015symmetry}, Restated]\label{theo:central-projection}
Let $\mf{g}$ be the DLA obtained from a set of generators $\GC$. If $[H,C]=0$ and $\Tr[HC]=0$ for every $H$ in $\GC$, then $\Tr[MC]=0$ for every $M$ in $\mf{g}$. In particular, for $\mf{g}$ to contain a nonzero central element $C$, it is necessary that $\Tr[HC]\neq 0$ for some $H$ in $\GC$.
\end{theorem}

Theorem~\ref{theo:central-projection} indicates that an element in the center cannot be in the DLA if all of the generators have zero projection onto it (in the sense that they are orthogonal with respect to the Frobenius inner product $\langle A,B\rangle = \Tr[A^\dagger B]$). Now we can use the following Theorem, proved in the Supplemental Material.

\begin{theorem}
\label{theo:central-projection-two-body}
Consider the set $\GC_2$ of $S_n$-equivariant generators in Eq.~\eqref{eq:genertors-2-body}. Then, one has that
\begin{equation}
    \Tr[H C_\mu]=0\,,
\end{equation}
for all $H\in\mathcal{G}_2$ and for all $C_\mu$ with $0\le \mu\le\lfloor n/2\rfloor$ except $\mu=1$.
\end{theorem}

Theorem~\ref{theo:central-projection-two-body} shows that the set of generators have no projection onto all center elements $C_\mu$ with $0\le \mu\le\lfloor n/2\rfloor$ except for that with $\mu=1$. Therefore, according to Theorem~\ref{theo:central-projection}, no $C_\mu$ with $0\le \mu\le\lfloor n/2\rfloor$ (except for that with $\mu=1$) can be generated within the DLA. It is important to remark that  the previous does not imply that $C_1$ actually appears in the DLA, as Theorem~\ref{theo:central-projection} only provides a necessary condition. However, one can prove by direct construction that $\mf{g}_2$ indeed  contains $C_1$, as well as any operator in $\mf{su}^{S_n}_{\rm{cless}}(d)$. That is, one can prove that the following theorem holds.

\begin{theorem}
\label{thm:true-AD-algebra}
Consider the set $\GC_2$ of $S_n$-equivariant generators in Eq.~\eqref{eq:genertors-2-body}. The associated DLA $\mf{g}_2$ contains all linear combinations of the form
\begin{equation}
    i\sum_{0\le k_x+k_y+k_z\le n}c_{(k_x,k_y,k_z)}P_{(k_x,k_y,k_z)},
\end{equation}
for a collection of real coefficients $c_{(k_x,k_y,k_z)}$ that satisfy
\begin{equation}
    \sum_{a+b+c=\mu}\frac{c_{(2a,2b,2c)}}{a!b!c!} = 0
\end{equation}
for each integer $0\le \mu\le\lfloor n/2\rfloor$ except $\mu=1$.
\end{theorem}

Here we note that Theorem~\ref{thm:true-AD-algebra} is simply a reformulation of Theorem~\ref{theo:subspace-not-universal} which allows us to directly identify $\mf{g}_2$ as containing all elements in $\mf{su}^{S_n}_{\rm{cless}}(d)$ plus the span of the center element $Q\left(\mf{u}(1)\right)=\spn_{\mbb{R}}\{iC_1\}$.  In addition, we can infer from Theorem~\ref{thm:true-AD-algebra} that the dimension of $\mf{g}_2$ is
\begin{equation}
    \text{dim}(\mf{g}_2) = \binom{n+3}{3} - \left\lfloor\frac{n}{2}\right\rfloor.
\end{equation}


Our previous results  proved that sets of generators with $S_n$-equivariant  $1$-body and $2$-body operators are not sufficient to generate $\mf{su}^{S_n}(d) $. A natural question then is whether this can be fixed by including in the set of $k$-body   $S_n$-equivariant generators (for $k\geq 3$). 
Defining the set of generators 
\begin{equation}\label{eq:set-gen-k}
    \GC_k=\{P_{(1,0,0)},P_{(0,1,0)}\}\cup\{P_{(0,0,\kappa)} \}_{\kappa=2}^k\,,
\end{equation}
we find that the following result also holds (see the Supplementary Information for a proof).
\begin{theorem}\label{theo:k-body-DLA}
Consider the set $\GC_k$ of $S_n$-equivariant generators in Eq.~\eqref{eq:set-gen-k}. Then, the associated DLA $\mf{g}_k$ is
\begin{align}
\mf{g}_k&=Q\left(\bigoplus_{\lambda=1}^L\mf{su}(m_\lambda)\oplus\underbrace{\mf{u}(1)\oplus\cdots\oplus\mf{u}(1)}_{\lfloor k/2\rfloor}\right)\\
&=\mf{su}^{S_n}_{\rm{cless}}(d) \boxplus Q\left( \underbrace{\mf{u}(1)\oplus\cdots\oplus\mf{u}(1)}_{\lfloor k/2\rfloor}  \right) \,,
\end{align}
where $Q\left( \mf{u}(1)\oplus\cdots\oplus\mf{u}(1)  \right) = \spn_{\mbb{R}}\{iC_1,\cdots,iC_{\lfloor k/2\rfloor}\}$ is a $\lfloor k/2\rfloor$-dimensional subalgebra of $\mf{z}(\mf{u}^{S_n}(d))$.

\end{theorem}

Theorem~\ref{theo:k-body-DLA} shows that one element in the center is generated in the DLA every time one adds a generator with \textit{even} bodyness. Crucially, this result is in accordance with the central projections condition of Theorem~\ref{theo:central-projection}. In particular, the central projections condition implies that one cannot generate $\mf{su}^{S_n}(d) $ unless one has, for every $\mu$ from $0$ to $\lfloor n/2\rfloor$, some generator with nonzero projection onto $C_\mu$.  However, it is easy to see that the following theorem holds.
\begin{theorem}
\label{theo:central-projection-k-body}
Consider the set $\GC_k$ of $S_n$-equivariant generators in Eq.~\eqref{eq:set-gen-k}. Then, one has that
\begin{equation}
    \Tr[H C_\mu]=0\,,
\end{equation}
for all $H\in\mathcal{G}_k$ and for all $C_\mu$ with $\mu=0$ and  $\lfloor \frac{k}{2}\rfloor< \mu\le\lfloor n/2\rfloor$.
\end{theorem}
Theorem~\ref{theo:central-projection-k-body} shows that the operators in $\GC_k$ have no projection onto $C_\mu$ with $\mu>\lfloor \frac{k}{2}\rfloor$ (or $\mu=0$). Thus, one finds  that the following corollary  holds.
\begin{corollary}\label{theo:k-body-conditions}
    Any set consisting of at-most-$k$-body $S_n$-equivariant operators will always be insufficient to generate $\mf{su}^{S_n}(d)$ unless $k=n$ for $n$ even or $k\geq n-1$ for $n$ odd.
\end{corollary}

It is worth noting that Corollary~\ref{theo:k-body-conditions} can also follow from the no-go theorem of Ref.~\cite{marvian2022restrictions}, along with the fact that the center $\mf{z}\left(\mf{u}^{S_n}(d)\right)$ contains $n$-body operators for $n$ even or $(n-1)$-body operators for $n$ odd.

Theorem~\ref{theo:k-body-DLA} and Corollary~\ref{theo:k-body-conditions} shows that in order for a QNN with $S_n$-equivariant $k$-body gates to be universal, one needs to include in the set of generators up-to-$n$-body gates (for $n$ even) or up-to-$(n-1)$-body gates (for $n$ odd). Hence, as schematically depicted in Fig.~\ref{fig:Summary}, this corollary imposes a fundamental limitation of the universality of QNNs with $S_n$-equivariant $k$-local gates.

\begin{figure}[t]
    \centering
\includegraphics[width=1\columnwidth]{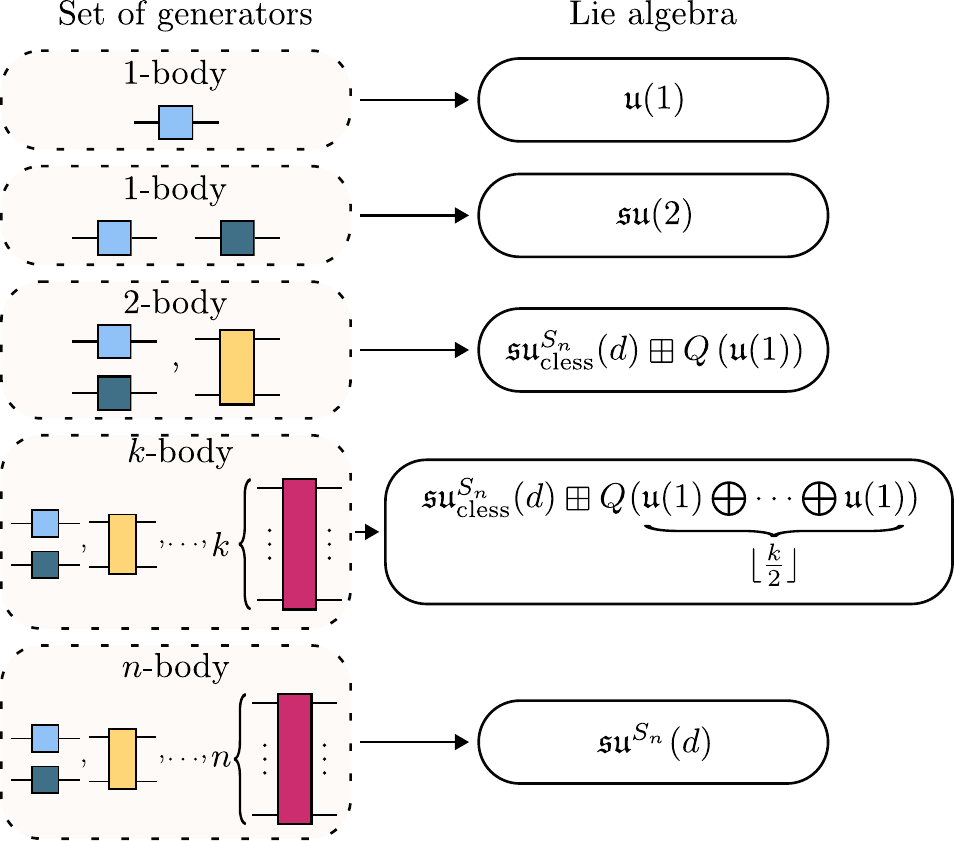}
    \caption{\textbf{$S_n$-equivariance, few-bodyness, and DLA.} Here we review the main results of our work.  In particular, we consider the case where the single-body operators in $\GC$ are  $P_{(1,0,0)}$ and  $P_{(0,1,0)}$ and where the $k$-body operators are of the form  $P_{(0,0,k)}$.
    \label{fig:Summary}}
\end{figure}

\section{Discussion}

The recent field of GQML has paved the way for studying how adding symmetry information into quantum learning models changes their performance in terms of expressiveness, trainability, and generalization. Recent results has shown that GQML can indeed provide a heuristic advantage for several machine learning tasks~\cite{west2023reflection,east2023all,Chang2023Approximately,zheng2023sncqa,wierichs2023symmetric,le2023symmetry,dong2024Z2,forestano2024comparison} over their symmetry agnostic counter-parts. In particular, it was shown that  the special case of  $S_n$-equivariant QNNs exhibit the holy grail of desirable properties~\cite{schatzki2022theoretical}: absence of barren plateaus, generalization from few training points, and the capacity to be efficiently overparametrized. Still, despite the great promise of GQML there are still many open questions regarding its full capabilities. As an example, the seminal works of Refs.~\cite{marvian2022restrictions,marvian2022rotationally,marvian2021qudit,marvian2023non}, have started to analyze restrictions to universality and their connections to the type of gates used. Moreover, it is has been documented that quantum noise can be quite detrimental for equivariant models~\cite{tuysuz2024symmetry,garcia2023effects,wang2020noise} and that reducing a quantum learning model's expressive power too much can potentially make it classically simulable~\cite{anschuetz2022efficient,cerezo2023does}. As such, we expect that the study of symmetry-enforced models will be a thriving area in the future.

In this context, we have shown that quantum circuits with elementary $k$-body $S_n$-equivariant gates are semi-universal and subspace controllable (see Definition~\ref{def:subspace-control-new}), but not universal. While this result might seem negative (as we cannot reach universality), we are inclined to note that reaching semi-universality is in itself a tall order. The fact that $S_n$-equivariant QNNs with only $2$-body gates have enough expressiveness to be semi-universal is a notable result in itself. Crucially, it is not obvious that such high degree of universality  will generally hold. Moreover, our results indicate that in order to reach universality, one must include up-to-$n$-body interactions if $n$ is even, and up-to-$(n-1)$-body interactions if $n$ is odd. These results have several implications. First, they highlight the existence of a fundamental limitation to achieving universality from local permutation-invariant gates. Second, they correct a result in the literature which states that one- and two-body $S_n$-equivariant gates are indeed enough for universality. Finally, they portray the power of the central projections condition of Ref.~\cite{zimboras2015symmetry}.

Here we also note that at a higher level, the central projections condition offers a significant obstruction to the possibility of using only local or few-body gates to achieve universality of gates that respect some symmetry group (see also~\cite{marvian2022rotationally}). Of course, symmetries are necessary for this phenomenon to present itself; for instance, as illustrated by a classical result in quantum computation~\cite{divincenzo1995two}, the collection of all $1$-body and $2$-body Pauli strings, which do not share any symmetry, is indeed sufficient to generate all Pauli strings.

It is also worth highlighting the fact that our results are intrinsically connected to those in Ref.~\cite{marvian2022rotationally} through the Schur-Weyl duality. Namely, while we study circuits with $S_n$-equivariance, the work in~\cite{marvian2022rotationally} studies $\mathbb{SU}(2)$-equivariant circuits. Since the qubit permuting representation of $S_n$ and the tensor product representation of $\mathbb{SU}(2)$ mutually centralize each other, their centers match. Hence, the conditions imposed by the central projections theorem  to the circuit's generators bodyness are exactly the same for $\mathbb{SU}(2)$- or $S_n$-equivariant circuits.  That is, Corollary~\ref{theo:k-body-conditions} in our work is precisely, and necessarily, the same as the direction of Theorem 1 in~\cite{marvian2022rotationally} that establishes constraints from the central projections condition. This realization then allows us to identify our work as being dual to that in Ref.~\cite{marvian2022rotationally}, where each work explores how the central projections conditions affect generators on each ``side'' of the Schur-Weyl duality. Finally, this connection shows that results obtained for one side of a dual reductive pair will naturally be relevant for the other side.

As a final note, both our work and those in 
Refs.~\cite{marvian2022restrictions,marvian2022rotationally} focus primarily on obstructions to universality coming from the central projections condition. However, the characterization of Lie algebras presented in Ref.~\cite{zimboras2015symmetry} has three components: linear symmetries, quadratic symmetries, and central projections, with central projections generally having the smallest impact in terms of reducing the dimension of the Lie algebra. As an example, the Lie algebra $\mf{g}_2$ studied in this paper has a relatively small dimension deficit of $\lfloor n/2\rfloor$ compared to the full $S_n$-invariant algebra $\mf{u}^{S_n}(d)$. The question of whether few-bodyness (or locality) of symmetry-equivariant qubit gates could enforce extra linear or quadratic symmetries has only begun to be studied, for instance in Refs.~\cite{marvian2021qudit,marvian2023non}, and such possibilities would cause the DLA dimension deficiency to be much larger. We thus believe that a more general investigation into the failure modes for gate sets with additional constraints (beyond symmetry-equivariance) to generate the full symmetry-invariant algebra will be a fruitful area of study.

\textit{Note added:} After completion of this manuscript, it was brought to our attention that Proposition~1 of our Supplemental Material is also presented in a slightly more general form as Proposition 1 in \cite{marvian2022rotationally}.

\section*{Acknowledgements}
We would like to thank Robert Zeier for useful discussions regarding his result on linear symmetries, quadratic symmetries, and central projections; as well as Iman Marvian for useful discussions regarding Schur-Weyl duality and dual reductive pairs. We also thank Domenico D'Alessandro for useful comments regarding our manuscript. SK was initially supported by the U.S. DOE through a quantum computing program sponsored by the Los Alamos National Laboratory (LANL) Information Science \& Technology Institute. 
ML was supported by the Center for Nonlinear Studies at LANL. ML and MC acknowledge  support by the Laboratory Directed Research and Development (LDRD) program of LANL  under project numbers 20230527ECR and 20230049DR.

\section*{Data availability statement}

No new data were created or analysed in this study.

\bibliography{quantum}

\makeatletter
\close@column@grid
\makeatother
\cleardoublepage
\newpage
\onecolumngrid
\appendix

\section*{Supplemental Material for ``\textit{On the universality of $S_n$-equivariant $k$-body gates}''}

Here we present the proofs and additional details for the results in the main text. First, in Supplementary Note~\ref{section:notation} we present preliminaries and review the notation used throughout this Supplementary Information. Then, in Supplementary Note~\ref{section:proof_Prop-1} we provide a proof for Proposition~1. 
In particular, we note that we will prove Theorem~1 by proving first Theorems~3 and~4 (see Supplementary Notes~\ref{sec:theo:central-projection-two-body} and~\ref{sec:thm:true-AD-algebra}). Supplementary note~\ref{sec:theo:k-body-DLA} contains the proof of Theorem~5, and Supplementary Note~\ref{sec:theo:central-projection-k-body} a proof of Theorem~6. Finally, Supplementary Note~\ref{sec:error} explains the bug in the proof of Ref.~\cite{albertini2018controllability} which claims that $\mf{g}_2$ is universal.

\renewcommand\appendixname{Supplementary Note}

\renewcommand\figurename{Supplementary Figure}
\setcounter{figure}{0}

\section{Review of Notation}
\label{section:notation}

Let us review the notation we use throughout this paper. First, we define the Pauli matrices:
\begin{equation}
    I = \begin{bmatrix} 1 & 0 \\ 0 & 1 \end{bmatrix} \quad
    X = \begin{bmatrix} 0 & 1 \\ 1 & 0 \end{bmatrix} \quad
    Y = \begin{bmatrix} 0 & -i \\ i & 0 \end{bmatrix} \quad
    Z = \begin{bmatrix} 1 & 0 \\ 0 & -1 \end{bmatrix}.
\end{equation}
Next, for Hermitian operators $H_1, H_2, H_3$, we will use the notation
\begin{equation}
    [H_1, H_2] \propto H_3\,,
\end{equation}
to signify the equation
\begin{equation}
    [iH_1, iH_2] = -2iH_3.
\end{equation}
We will use this notation to avoid writing factors of $2i$ when taking the commutators of Pauli operators. For example, we can write $[X,Y]\propto Z$ in place of $[iX,iY] = -2iZ$.

\vspace{0.5\baselineskip}

Furthermore, we will typically abuse notation and say that $H\in\mf{g}$ for a Hermitian operator $H$ and a Lie algebra $\mf{g}$, when what we really mean is $iH\in\mf{g}$, since we always deal with Lie algebras of skew-Hermitian operators unless we explicitly state otherwise.

\vspace{0.5\baselineskip}

Now we recall the definition of $S_n$-symmetrized Pauli strings.

\begin{definition}
Define $P_{(k_x,k_y,k_z)}$ to denote the sum of all distinct Pauli strings that have $k_x$ $X$ symbols, $k_y$ $Y$ symbols, and $k_z$ $Z$ symbols. Furthermore, define the \textbf{level} of this operator to be the total number of non-$I$ symbols, i.e. $k_x + k_y + k_z$.
\end{definition}

In Sup. Fig.~\ref{fig:proof_1} we present a useful visual representation for these operators in terms of \textbf{barycentric coordinates}. Here, we define the barycentric coordinates as $k_x$, $k_y$, $k_z$, and the leftover as $k_i = n-(k_x+k_y+k_z)$. This means that the operators will be assembled into a regular tetrahedron, where for each integer $0\le k\le n$, the operators at level $k$ form an equilateral triangle at elevation $k$ (see the  tetrahedral stack of spheres in Sup. Fig.~\ref{fig:proof_1}(b)).

\vspace{0.5\baselineskip}

Finally, we recall that the $P_{(k_x,k_y,k_z)}$ operators can be used to define some of the relevant algebras we care for:

\begin{align}
\begin{split}\label{eq:maximal-algebra-P-ap}
\mf{u}^{S_n}(d) = \spn_{\mbb{R}}\{iP_{(k_x,k_y,k_z)} \,|\, k_x+k_y+k_z\le n\}\,,\\
\mf{su}^{S_n}(d) = \spn_{\mbb{R}}\{iP_{(k_x,k_y,k_z)} \,|\, 0< k_x+k_y+k_z\le n\}.
\end{split}
\end{align}
Using the visual aide based on barycentric coordinates in Sup. Fig.~\ref{fig:proof_1}, we see that the dimension of $\mf{u}^{S_n}(d)$ is equal to the $(n+1)^\text{th}$ tetrahedral number $\text{Te}_{n+1} = \binom{n+3}{3}$. There are many ways to prove the equality with the binomial coefficient, but one of the most intuitive ones is to use the ``stars and bars'' combinatorial argument. In particular, $\text{dim}\left(\mf{u}^{S_n}(d)\right)$ equals the number of ordered triples $(k_x,k_y,k_z)$ of non-negative integers such that $k_x+k_y+k_z\le n$. This equals the number of ways to line up $n+3$ objects and choose $3$ of them to be the ``boundaries'', such that $k_x$ equals the number of elements to the left of the first boundary, $k_y$ equals the number of elements between the first and second boundaries, $k_z$ equals the number of elements between the second and third boundaries (the leftover $n-(k_x+k_y+k_z)$ equals the number of elements to the right of the third boundary).

\begin{figure}[t]
    \centering
\includegraphics[width=1\columnwidth]{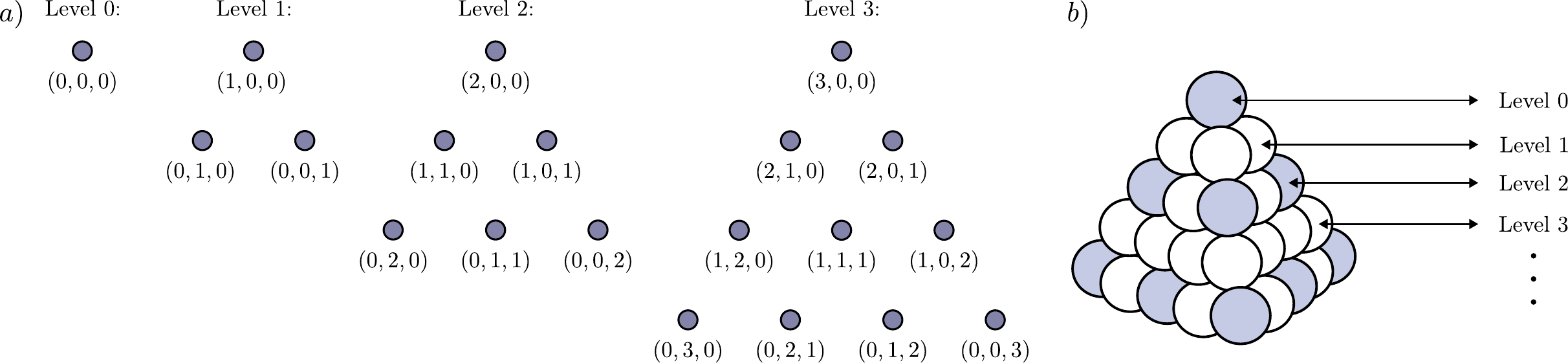}
    \caption{\textbf{Barycentric coordinates.} a) Here we show the barycentric coordinates for $n=3$, where the levels $0$, $1$, $2$ and $3$ shown separately for ease of visualization. b) One can imagine these levels stacked on top of each other to form a regular tetrahedron.
    \label{fig:proof_1}}
\end{figure}

\section{Proof of Proposition~1}\label{section:proof_Prop-1}

Here we provide a proof for Proposition~1, which we recall for convenience.
\setcounter{proposition}{0}
\begin{proposition}\label{lem:Ck-span-Sn-center-app} 
The center of the algebra $\mf{u}^{S_n}(d)$ is the span of the $C_\mu$ operators:
\begin{equation}
    \mf{z}(\mf{u}^{S_n}(d)) = \spn_{\mbb{R}}\Big\{i C_\mu \,|\, 0\le \mu\le \Big\lfloor\frac{n}{2}\Big\rfloor \Big\}.
\end{equation}
\end{proposition}

Here we also recall the definition of the $C_\mu$ operators.
For each integer $0\le \mu\le\lfloor n/2\rfloor$, we have
\begin{equation}
    C_\mu = \sum_{a+b+c=\mu}\frac{(2a)!(2b)!(2c)!}{a!b!c!}P_{(2a,2b,2c)}.
\end{equation}
Each $C_\mu$ consists exclusively of $P_{(k_x,k_y,k_z)}$ terms where $k_x, k_y, k_z$ are all even and add to $2\mu$ (which also implies that each $C_\mu$ lives entirely on level $2\mu$, see Sup. Fig.~\ref{fig:proof_1}(a)).

A standard result in representation theory known as the \textit{Schur-Weyl duality}~\cite{goodman2009symmetry} states, in the context of a $n$-qudit quantum system, that
the representation $R$ of $S_n$ that permutes the qubits and the tensor product representation $T$ of $\mbb{U}(2)$ that acts locally and identically on each qubit, defined by $T(U\in \mbb{U}(2))=U\tn$, mutually centralize each other. That is,
\small
\begin{align}
&\mf{u}^{\mbb{U}(2)}(d) =  \spn_{\mbb{R}}\{iR (\pi) \,|\, \pi\in S_n\} = \mf{S}_n\\
&\mf{u}^{S_n}(d) = \spn_{\mbb{R}}\{ i(T(U)+T(U)\ad) \,|\, U\in \mbb{U}(2)\}=\mf{U}(2) \,,
\end{align}
\normalsize
where we use $\mf{S}_n$ and $\mf{U}(2)$ to denote the Lie algebras of qubit permutations and of $\mbb{U}(2)$ tensor-product rotations, respectively. Crucially, the previous equations highlight the fact that the $\mbb{U}(2)$-symmetric part of $\mf{u}(d)$ is given by qubit permutations and that the $S_n$-symmetric part of $\mf{u}(d)$ is given by $\mbb{U}(2)$ rotations. 

Instead of directly proving Proposition~\ref{lem:Ck-span-Sn-center-app} we will prove the following equivalent statement.
\begin{proposition}\label{lem:alternative}
The center of the algebra $\mf{u}^{\mbb{U}(2)}(d)$ is the span of the $C_\mu$ operators:
\begin{equation}
    \mf{z}(\mf{u}^{\mbb{U}(2)}(d)) = \spn_{\mbb{R}}\Big\{i C_\mu \,|\, 0\le \mu\le \Big\lfloor\frac{n}{2}\Big\rfloor \Big\}.
\end{equation}
\end{proposition}

Crucially, Proposition~\ref{lem:alternative} and Proposition~\ref{lem:Ck-span-Sn-center-app} are equivalent since being mutual centralizers, their centers must necessarily coincide: $\mf{z}(\mf{u}^{S_n}(d))=\mf{z}(\mf{u}^{\mbb{U}(2)}(d))$ . Hence, proving Proposition~\ref{lem:alternative} is the same as proving Proposition~\ref{lem:Ck-span-Sn-center-app}.

We begin by citing a general characterization of the center of any algebra defined by a finite group representation:

\begin{lemma}
\label{lem:center-spanned-by-conj-class-sums}
Consider a (Hermitian unitary) representation $S$ of a finite group $G\subseteq S_n$, and consider the algebra $\mf{h} = \spn\{iS(g) \,|\, g\in G\}$. The center of this algebra, denoted $\mf{z}(\mf{h})$, is exactly the span of the conjugacy-class sums of $G$. More precisely, if $C_1,\cdots,C_p$ are the conjugacy classes of $G$, and $S(C_j) = \sum_{g\in C_j}S(g)$ are the conjugacy-class sums, then
\begin{equation}
    \mf{z}(\mf{h}) = \spn\{iS(C_j) \,|\, 1\le j\le p\}.
\end{equation}
\end{lemma}

\begin{proof}
See Ref. \cite{serre1977linear}, more specifically Proposition 6 in Section 2.5 and Proposition 12 in Section 6.3.
\end{proof}

Lemma \ref{lem:center-spanned-by-conj-class-sums} motivates us to define the following set of operators, each of which is a conjugacy-class sum within $\mf{S}_n$:

\begin{definition}
For each integer $0\le \mu\le\lfloor n/2\rfloor$, define $T_\mu$ to be the set of all permutations in the symmetric group $S_n$ that consist of exactly $\mu$ disjoint transpositions. Then define the following operator:
\begin{equation}
    L_\mu = \sum_{\pi\in T_\mu}R(\pi)\,.
\end{equation}
\end{definition}
As a quick example, consider the operator obtained from two element transpositions
\begin{align}
    L_1 &= \sum_{1\le \alpha <\beta \le n}R(\pi=(\alpha \,\beta)) \\
    &= \sum_{1\le \alpha<\beta\le n}\frac{1}{2}(I_\alpha I_\beta + X_\alpha X_\beta + Y_\alpha Y_\beta + Z_\alpha Z_\beta) \\
    &= \frac{1}{2}\binom{n}{2}P_{(0,0,0)} + \frac{1}{2}P_{(2,0,0)} + \frac{1}{2}P_{(0,2,0)} + \frac{1}{2}P_{(0,0,2)}.
\end{align}

Now, armed with Lemma~\ref{lem:center-spanned-by-conj-class-sums}, we can provide a proof for the following lemma (see also~\cite{zheng2021speeding} for an alternative proof).

\begin{lemma}
\label{lem:Lk-span-Sn-center}
The center of the algebra $\mf{S}_n$ is the span of the $L_\mu$ operators:
\begin{equation}
    \mf{z}(\mf{S}_n) = \spn\Big\{i L_\mu \,|\, 0\le \mu\le \Big\lfloor\frac{n}{2}\Big\rfloor \Big\}.
\end{equation}
\end{lemma}

\begin{proof}
Each conjugacy class of the symmetric group is the set of all permutations of a given cycle type. Hence each $T_\mu$ is a conjugacy class. And since $L_\mu$ is a sum over all permutation operators in this conjugacy class, Lemma \ref{lem:center-spanned-by-conj-class-sums} immediately implies that $L_\mu$ commutes with all permutation operators. Hence each $L_\mu$ lies in the center of $\mf{S}_n$.

Now it remains to show that these are all of the central elements of $\mf{S}_n$. Clearly the $T_\mu$ are only a small number of the conjugacy classes of $S_n$, so Lemma~\ref{lem:center-spanned-by-conj-class-sums} will not help us with that. However, it is known by an argument involving Young diagrams that the dimension of the center of $\mf{S}_n$ is given by the number of Young diagrams with $n$ cells that have at most $2$ rows~\cite{goodman2000representations}. Such a diagram has $\mu$ cells in the second row and $n-\mu$ cells in the first row for some integer $0\le \mu\le\lfloor\frac{n}{2}\rfloor$, so there are exactly $\lfloor\frac{n}{2}\rfloor + 1$ such diagrams.  Hence, this means that there are $\lfloor\frac{n}{2}\rfloor + 1$ elements in the center which coincides with the number of $L_\mu$ operators. Thus, $L_\mu$ span the center of the algebra $\mf{S}_n$.
\end{proof}

As a brief aside, it is worth mentioning that the theorem that we cited above regarding Young diagrams depends on the fact that we are working with qubits. For example, if we were working with qutrits instead of qubits, then one would instead need to look at Young diagrams with at most $3$ rows~\cite{goodman2000representations}.

Up to this point, we have shown that the $L_\mu$ operators span $\mf{z}(\mf{S}_n)$, so all that remains is to show that the $C_\mu$ operators actually have the same span as the $L_\mu$ operators. We demonstrate this now via the following lemma.

\begin{lemma}
\label{lem:Lk-in-terms-of-Ck}
For each integer $0\le \mu\le\lfloor n/2\rfloor$,
\begin{equation}
    L_\mu = \frac{4^{-\mu}}{(n-2k)!}\sum_{a+b+c+f=\mu}\frac{(2a)!(2b)!(2c)!(n-2(\mu-f))!}{a!b!c!f!}P_{(2a,2b,2c)}.
\end{equation}
As a result, we can also write
\begin{equation}
    L_\mu = \frac{4^{-\mu}}{(n-2\mu)!}\sum_{\mu'=0}^{\mu}\frac{(n-2\mu')!}{(\mu-\mu')!}C_{\mu'}.
\end{equation}
\end{lemma}

\begin{corollary}
\label{cor:Lk-Ck-same-span}
For each integer $0\le \mu\le\lfloor n/2\rfloor$, the $L_{\mu}$ operators and the $C_{\mu'}$ operators for $0\le \mu'\le \mu$ have the same span. In particular, the full set of $(\lfloor n/2\rfloor + 1)$ $L_\mu$ operators has the same span as the full set of $(\lfloor n/2\rfloor + 1)$ $C_\mu$ operators.
\end{corollary}

\begin{proof}
The transposition $(\alpha\,\beta)$ has qubit representation $R(\pi={(\alpha\,\beta)}) = \frac{1}{2}(I_\alpha I_\beta + X_\alpha X_\beta + Y_\alpha Y_v + Z_\alpha Z_\beta)$, and an element $\pi\in T_\mu$ takes the form $\pi = (\alpha_1\,\beta_1)\cdots(\alpha_\mu\,\beta_\mu)$ for distinct $\alpha_1,\cdots,\alpha_\mu,v_1,\cdots,v_\mu\in\{1,\cdots,n\}$. Therefore,
\begin{equation}
    R(\pi) = 2^{-\mu}\prod_{j=1}^{\mu}(I + X_{\alpha_j}X_{\beta_j} + Y_{\alpha_j}Y_{\beta_j} + Z_{\alpha_j}Z_{\beta_j})\,,
\end{equation}
which after expansion implies that $R(\pi)$ consists solely of Pauli strings with an even number of $X$ symbols, an even number of $Y$ symbols, and an even number of $Z$ symbols. Hence the same is true of $L_\mu$, and since $L_\mu$ is $S_n$-invariant, we know that $L_\mu$ is a linear combination of symmetrized Pauli strings $P_{(2a,2b,2c)}$.

The coefficient of $P_{(2a,2b,2c)}$ in $L_\mu$ equals $2^{-\mu}$ times the number of permutations $\pi\in T_\mu$ for which any individual Pauli string within $P_{(2a,2b,2c)}$ appears in the expansion of $R(\pi)$. So, consider an arbitrary Pauli string with $2a$ $X$'s, $2b$ $Y$'s, $2b$ $Z$'s, and $n-2(a+b+c)$ $I$'s. The $2a$ qubits with $X$'s must be paired up into $a$ transpositions within $\pi$. The number of ways to choose how they are paired up is
\begin{equation}
    \frac{1}{a!}\binom{2a}{2}\binom{2a-2}{2}\cdots\binom{2}{2} = \frac{(2a)!}{2^aa!}.
\end{equation}
The same is true for the $2b$ qubits with $Y$'s, and the same is also true for the $2c$ qubits with $Z$'s. Hence we additionally get factors of $\frac{(2b)!}{2^bb!}$ and $\frac{(2c)!}{2^cc!}$. The remaining freedom in defining $\pi$ comes from the number of ways to pair up some of the $n-2(\mu-f)$ qubits with $I$'s into the remaining $d$ transpositions, where we have defined $f=\mu-(a+b+c)$ for convenience. The number of ways to do this is
\begin{equation}
    \frac{1}{f!}\binom{n-2(\mu-f)}{2}\binom{n-2(\mu-f)-2}{2}\cdots\binom{n-2\mu+2}{2} = \frac{(n-2(\mu-f))!}{2^ff!(n-2\mu)!}.
\end{equation}
Putting this all together, we obtain that the coefficient of the Pauli string in question within $L_\mu$ is
\begin{equation}
    2^{-\mu}\frac{(2a)!}{2^aa!}\frac{(2b)!}{2^bb!}\frac{(2c)!}{2^cc!}\frac{(n-2(\mu-f))!}{2^ff!(n-2\mu)!} = \frac{4^{-\mu}}{(n-2\mu)!}\frac{(2a)!(2b)!(2c)!(n-2(\mu-f))!}{a!b!c!f!}.
\end{equation}
We conclude that
\begin{equation}
    L_\mu = \frac{4^{-\mu}}{(n-2\mu)!}\sum_{a+b+c+f=\mu}\frac{(2a)!(2b)!(2c)!(n-2(\mu-f))!}{a!b!c!f!}P_{(2a,2b,2c)},
\end{equation}
which completes the first part of this proof.

\vspace{0.5\baselineskip}

Now, for the second part of this proof, simply let $\mu'=a+b+c=\mu-f$ and regroup the summation as follows:
\begin{align}
    L_\mu &= \frac{4^{-\mu}}{(n-2\mu)!}\sum_{a+b+c+f=\mu}\frac{(2a)!(2b)!(2c)!(n-2(\mu-f))!}{a!b!c!f!}P_{(2a,2b,2c)} \\
    &= \frac{4^{-\mu}}{(n-2\mu)!}\sum_{f=0}^{\mu}\frac{(n-2(\mu-f))!}{f!}\sum_{a+b+c=\mu-f}\frac{(2a)!(2b)!(2c)!}{a!b!c!}P_{(2a,2b,2c)} \\
    &= \frac{4^{-\mu}}{(n-2\mu)!}\sum_{f=0}^{\mu}\frac{(n-2(\mu-f))!}{f!}C_{\mu-f} \\
    &= \frac{4^{-\mu}}{(n-2\mu)!}\sum_{\mu'=0}^{\mu}\frac{(n-2\mu')!}{(\mu-\mu')!}C_{\mu'}.
\end{align}
\end{proof}

To summarize, we now know from Lemma \ref{lem:Lk-span-Sn-center} that the $L_\mu$ operators span $\mf{z}(\mf{S}_n)$, and we also know from Corollary \ref{cor:Lk-Ck-same-span} (an immediate consequence of Lemma \ref{lem:Lk-in-terms-of-Ck}) that the $L_\mu$ operators and the $C_\mu$ operators have the same span. We conclude that the $C_\mu$ operators span $\mf{z}(\mf{S}_n)$, which is exactly the statement of Proposition~\ref{lem:alternative}. In addition, since  $\mf{z}(\mf{u}^{S_n}(d))=\mf{z}(\mf{S}_n)$, we can also conclude that the $C_\mu$ operators span $\mf{z}(\mf{u}^{S_n}(d))$, as stated in Proposition~\ref{lem:Ck-span-Sn-center-app}.

After completion of this manuscript, it was brought to our attention that Proposition \ref{lem:Ck-span-Sn-center-app} in this manuscript is also presented in a slightly more general form as Proposition 1 in \cite{marvian2022rotationally}. Some of the lemmas used to prove this proposition, namely Lemma \ref{lem:Lk-span-Sn-center}, Lemma \ref{lem:Lk-in-terms-of-Ck}, and Corollary \ref{cor:Lk-Ck-same-span}, are also presented in \cite{marvian2022rotationally}. The $L_\mu$ and $C_\mu$ operators in this manuscript correspond respectively to the $B_m$ and $C_l$ operators in \cite{marvian2022rotationally}, with $m = 2\mu$ and $l = 2\mu$.

\section{Proof of Theorem~3}\label{sec:theo:central-projection-two-body}

Let us now prove Theorem~3, which we restate for convenience.

\setcounter{theorem}{2}
\begin{theorem}
Consider the set $\GC_2$ of $S_n$-equivariant generators in Eq.~(23) of the main text. Then, one has that
\begin{equation}
    \Tr[H C_\mu]=0\,,
\end{equation}
for all $H\in\mathcal{G}_2$ and for all $C_\mu$ with $0\le \mu\le\lfloor n/2\rfloor$ except $\mu=1$.
\end{theorem}

\begin{proof}
None of the generators in $\GC_2$ share any Pauli strings in common with any $C_\mu$, with the exception of $P_{(0,0,2)}$ which shares a Pauli strings of the form $Z_\alpha Z_\beta$ with $C_1$. In particular, for any $0\le \mu\le\lfloor n/2\rfloor$ such that $\mu\neq 1$, the generators $P_{(1,0,0)}$, $P_{(0,1,0)}$, and $P_{(0,0,2)}$ do not share any Pauli strings with $C_\mu$. Therefore, $P_{(1,0,0)}C_\mu$ is a linear combination of non-identity Pauli strings. Since all non-identity Pauli strings are traceless, $\Tr[P_{(1,0,0)}C_\mu] = 0$. The same is true if $P_{(1,0,0)}$ is replaced with either of the other two generators.
\end{proof}

The importance of Theorem~3 comes from the central projections condition. In particular, Theorem~3 and Theorem~2 immediately give us the following:

\begin{corollary}
\label{cor:central-projection-two-body-ap}
For every $M\in\mf{g}_2$ and every integer $0\le \mu\le\lfloor n/2\rfloor$ except $\mu=1$, $\Tr[MC_\mu] = 0$.
\end{corollary}

\section{Proof of Theorem~4}\label{sec:thm:true-AD-algebra}

Let us now prove Theorem~4, which we restate for convenience.

\begin{theorem}
\label{thm:true-AD-algebra-ap}
Consider the set $\GC_2$ of $S_n$-equivariant generators in Eq.~(23) of the main text. The associated  DLA  contains all linear combinations of the form
\begin{equation}
    i\sum_{0\le k_x+k_y+k_z\le n}c_{(k_x,k_y,k_z)}P_{(k_x,k_y,k_z)},
\end{equation}
for a collection of real coefficients $c_{(k_x,k_y,k_z)}$ that satisfy
\begin{equation}
    \sum_{a+b+c=\mu}\frac{c_{(2a,2b,2c)}}{a!b!c!} = 0\,,
\end{equation}
for each integer $0\le \mu\le\lfloor n/2\rfloor$ except $\mu=1$.
\end{theorem}

Let us start by noting that since the number of terms in $P_{(k_x,k_y,k_z)}$ equals the multinomial coefficient
\begin{equation}
    \binom{n}{k_x, k_y, k_z, n - (k_x + k_y + k_z)} = \frac{n!}{k_x!k_y!k_z!(n - (k_x + k_y + k_z))!},
\end{equation}
the projection of $P_{(k_x,k_y,k_z)}$ onto $C_k$ equals
\begin{equation}
    \frac{n!}{(2a)!(2b)!(2c)!(n-(2a+2b+2c))!}\frac{(2a)!(2b)!(2c)!}{a!b!c!} = \frac{n!}{(n-2\mu)!}\frac{1}{a!b!c!}\,,
\end{equation}
if $(k_x,k_y,k_z) = (2a,2b,2c)$ such that $a+b+c=2\mu$, and $0$ otherwise. Therefore, expressing any operators $M$ in the DLA $\mf{g}$ as  a linear combination of the form
\begin{equation}
    M = \sum_{0\le k_x+k_y+k_z\le n}c_{(k_x,k_y,k_z)}P_{(k_x,k_y,k_z)}\,,
\end{equation}
and using the condition $\Tr[MC_\mu] = 0$ for each integer $0\le \mu\le\lfloor n/2\rfloor$ except $\mu=1$ (from Corollary~\ref{cor:central-projection-two-body-ap}), after factoring out a term $\frac{n!}{(n-2\mu)!}$, yields the constraint
\begin{equation}
    \sum_{a+b+c=\mu}\frac{c_{(2a,2b,2c)}}{a!b!c!} = 0\,.
\end{equation}

We must now show that we can in fact generate all operators that satisfy the aforementioned conditions. In what follows, we will follow essentially the same process as that used in Ref.~\cite{albertini2018controllability}. Given the length of the proof, we find it convenient to  summarize the main steps  before diving in:
\begin{itemize}
    \item First, we establish some preliminary results that will make the construction easier. As just one example, we prove that, if $P_{(k_x,k_y+1,k_z-1)}\in\mf{g}$ and $P_{(k_x,k_y,k_z)}\in\mf{g}$, then $P_{(k_x,k_y-1,k_z+1)}\in\mf{g}$ (as long as $k_y,k_z\ge 1$, of course).
    \item Second, we prove that, if $P_{(k-1,1,0)}\in\mf{g}$, then $P_{(k_x,k_y,k_z)}\in\mf{g}$ for all $(k_x,k_y,k_z)$ such that $k_x + k_y + k_z = k$ and at least one of $k_x,k_y,k_z$ is odd. Notice that, if $k$ is odd, then at least one of $k_x,k_y,k_z$ is guaranteed to be odd, so then we can fill out the entirety of level $k$.
    \item Third, we show that $P_{(k-1,1,0)}\in\mf{g}$ for all $k$. In combination with the previous result, this lets us fill out all of the odd levels and all of the even levels except for operators of the form $P_{(2a,2b,2c)}$.
    \item Fourth, we prove that we can generate $\sum_{a+b+c=k}c_{(2a,2b,2c)}P_{(2a,2b,2c)}$ for any set of coefficients $c_{(2a,2b,2c)}$ such that $\sum_{a+b+c=k}\frac{c_{(2a,2b,2c)}}{a!b!c!}=0$.
\end{itemize}

\vspace{0.5\baselineskip}

\begin{figure}[t]
    \centering
\includegraphics[width=.5\columnwidth]{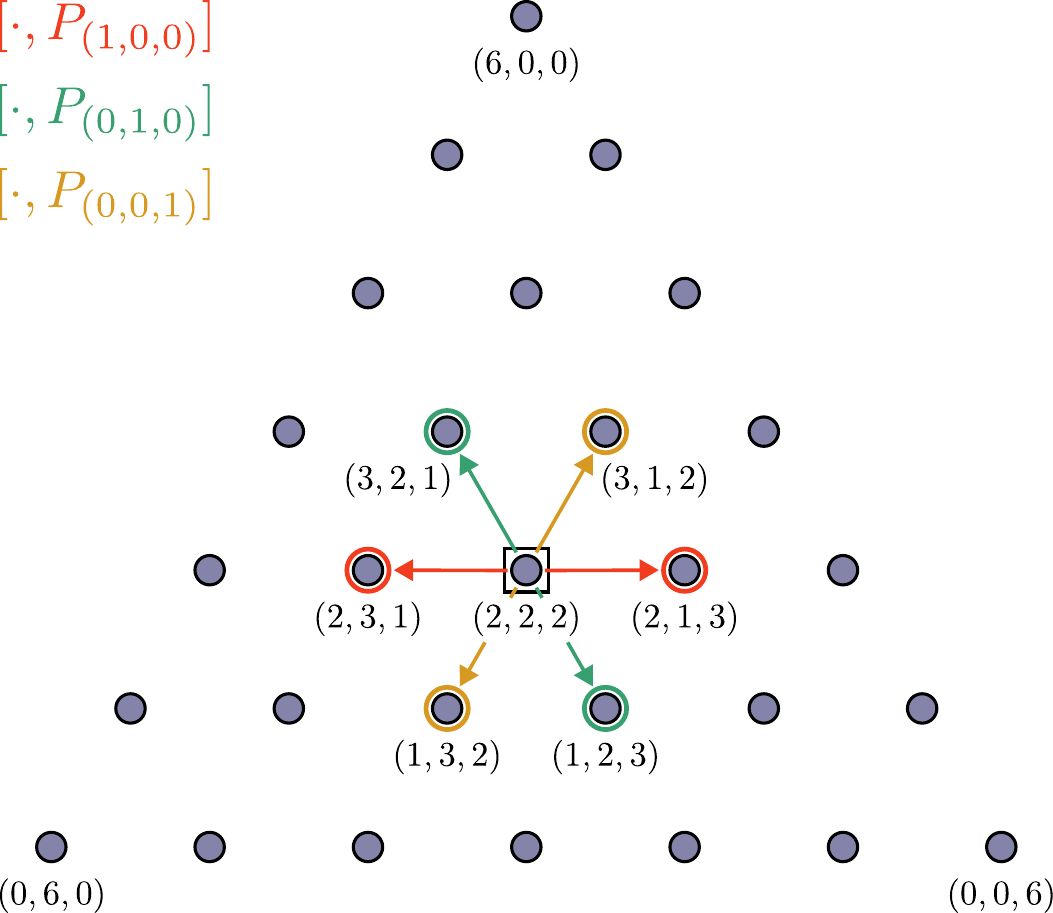}
    \caption{\textbf{Hopping by one site.} In the figure we schematically show how one can hop in the barycentric lattice by one site within the same level by taking specific  commutators. Shown is the  $k=6$ level. The red (horizontal) arrows show what happens when  taking the commutator of a symmetrized Pauli with $P_{(1,0,0)}=\sum_{j=1}^n X_j$. Similarly, the green (top-left to bottom-right) and the orange (bottom-left to top-right) arrows respectively show the connections made when taking the commutator with $P_{(0,1,0)}=\sum_{j=1}^n Y_j$ and $P_{(0,0,1)}=\sum_{j=1}^n Z_j$.
    \label{fig:proof_2}}
\end{figure}

We begin with an observation that, while seemingly trivial, is still worth mentioning. We note that in the following lemmas,  when we say $\mf{g}$, we will assume a DLA containing $P_{(1,0,0)}$, $P_{(0,1,0)}$, $P_{(0,0,1)}$ and $P_{(0,0,2)}$. That is, $\mf{g}_2\subseteq \mf{g}$.
\begin{lemma}
\label{lem:invariant-under-XYZ-interchange}
The DLA $\mf{g}_2$ is invariant under interchanges of the Pauli matrices $X,Y,Z$.
\end{lemma}

\begin{proof}
We can quickly carry out the following commutators that establish that $P_{(0,0,1)}$, $P_{(2,0,0)}$, $P_{(0,2,0)}$, $P_{(1,1,0)}$, $P_{(1,0,1)}$, and $P_{(0,1,1)}$ are all in $\mf{g}_2$:
\begin{align}
    \left[P_{(1,0,0)}, P_{(0,1,0)}\right] &\propto P_{(0,0,1)} \\
    \left[P_{(0,0,2)}, P_{(1,0,0)}\right] &\propto P_{(0,1,1)} \\
    \left[P_{(0,0,2)}, P_{(0,1,0)}\right] &\propto -P_{(1,0,1)} \\
    \left[P_{(0,1,1)}, P_{(1,0,0)}\right] &\propto P_{(0,2,0)} - P_{(0,0,2)} \\
    \left[P_{(1,0,1)}, P_{(0,1,0)}\right] &\propto -P_{(2,0,0)} + P_{(0,0,2)} \\
    \left[P_{(2,0,0)}, P_{(0,0,1)}\right] &\propto -P_{(1,1,0)}.
\end{align}
Hence the generators of $\mf{g}_2$ can be taken to be the set of operators $P_{(k_x,k_y,k_z)}$ for all $1\le k_x + k_y + k_z\le 2$. This generating set is invariant under interchanges of the Pauli matrices $X,Y,Z$, so the same must be true for $\mf{g}_2$ itself.
\end{proof}

We now present two lemmas that will help us greatly in systematically constructing the operators we claim to be in a given $\mf{g}$. The basic idea is that, under certain conditions, one can "hop" by one or two spaces within the barycentric lattice of level-$k$ symmetrized Pauli strings (see Sup. Fig.~\ref{fig:proof_2}). We first present the lemma that allows one to hop by one space.

\begin{lemma}
\label{lem:lattice-hop-one}
If $P_{(k_x,k_y,k_z)}\in\mf{g}$ and $P_{(k_x,k_y+1,k_z-1)}\in\mf{g}$, then $P_{(k_x,k_y-1,k_z+1)}\in\mf{g}$ (assuming, of course, that $k_y,k_z\ge 1$). If $P_{(k_x,k_y,0)}\in\mf{g}$ for $k_y\ge 1$, then $P_{(k_x,k_y-1,1)}\in\mf{g}$. Finally, if $P_{(k_x,0,0)}\in\mf{g}$ for $k_x\ge 1$, then $P_{(k_x-1,1,0)}\in\mf{g}$.
\end{lemma}

\begin{proof}
If $k_y,k_z\ge 1$, then simply use the commutator
\begin{equation}
    [P_{(k_x,k_y,k_z)}, P_{(1,0,0)}] \propto (k_y+1)P_{(k_x,k_y+1,k_z-1)} - (k_z+1)P_{(k_x,k_y-1,k_z+1)}\,,
\end{equation}
where we have used the fact that $P_{(1,0,0)}$ is in $\mf{g}$. Then,  since $P_{(k_x,k_y+1,k_z-1)}\in\mf{g}$, we can see that $P_{(k_x,k_y-1,k_z+1)}$ must also be in $\mf{g}$. If $k_z = 0$ but $k_y\ge 1$, then just use the commutator
\begin{equation}
    [P_{(k_x,k_y,0)}, P_{(1,0,0)}] \propto -(k_z+1)P_{(k_x,k_y-1,1)}\,.
\end{equation}
If $k_y = k_z = 0$, then just use the commutator
\begin{equation}
    [P_{(k_x,0,0)}, P_{(0,0,1)}] \propto -P_{(k_x-1,1,0)}\,.
\end{equation}
\end{proof}

Note that, although the lemma focuses on hopping one space only in certain directions, the lemma stills work regardless of which direction one hops, so long as one stays on level $k$. For instance, as schematically shown in Sup. Fig.~\ref{fig:proof_3}, taking the commutator with $P_{(1,0,0)}$ corresponds to hopping in the direction where one keeps the $x$-coordinate constant, but one increases (decreases) the $y$-coordinate by $1$ and respectively decreases (increases) the $z$-coordinate by $1$. Similarly, taking the commutator with $P_{(0,1,0)}$ corresponds to hopping in the direction where one keeps the $y$-coordinate constant, and taking the commutator with $P_{(0,0,1)}$ corresponds to hopping in the direction where one keeps the $z$-coordinate constant (see Sup. Fig.~\ref{fig:proof_3}). 

\medskip

We now present the lemma that lets us hop by two spaces. This lemma follows in spirit the same idea as Lemma \ref{lem:lattice-hop-one}, except now one  takes the commutator with $P_{(1,0,0)}$, $P_{(0,1,0)}$, or $P_{(0,0,1)}$ twice.

\begin{lemma}
\label{lem:lattice-hop-two}
If $P_{(k_x,k_y,k_z)}\in\mf{g}$ and $P_{(k_x,k_y+2,k_z-2)}\in\mf{g}$, then $P_{(k_x,k_y-2,k_z+2)}\in\mf{g}$ (assuming, of course, that $k_y,k_z\ge 2$). If $P_{(k_x,k_y,k_z)}\in\mf{g}$ for $k_y\ge 2$ and $k_z\le 1$, then $P_{(k_x,k_y-2,k_z+2)}\in\mf{g}$.
\end{lemma}

\begin{proof}
First, let us compute the commutator of $P_{(k_x,k_y,k_z)}$ with $P_{(1,0,0)}$ twice, i.e., $[[P_{(k_x,k_y,k_z)}, P_{(1,0,0)}], P_{(1,0,0)}]$. When one takes the first commutator, one gets a linear combination of $P_{(k_x,k_y+1,k_z-1)}$ and $P_{(k_x,k_y-1,k_z+1)}$. Then, when one takes the second commutator, one gets a linear combination of $P_{(k_x,k_y+2,k_z-2)}$, $P_{(k_x,k_y,k_z)}$, and $P_{(k_x,k_y-2,k_z+2)}$. (The exact coefficients are not important, which is why we choose not to emphasize them. The only important part is that the coefficient of $P_{(k_x,k_y-2,k_z+2)}$ is  not zero.) Since we are given that $P_{(k_x,k_y,k_z)}$ and $P_{(k_x,k_y+2,k_z-2)}$ are in $\mf{g}$, we conclude that $P_{(k_x,k_y-2,k_z+2)}\in\mf{g}$, as desired.

\begin{figure}[t]
    \centering
\includegraphics[width=.5\columnwidth]{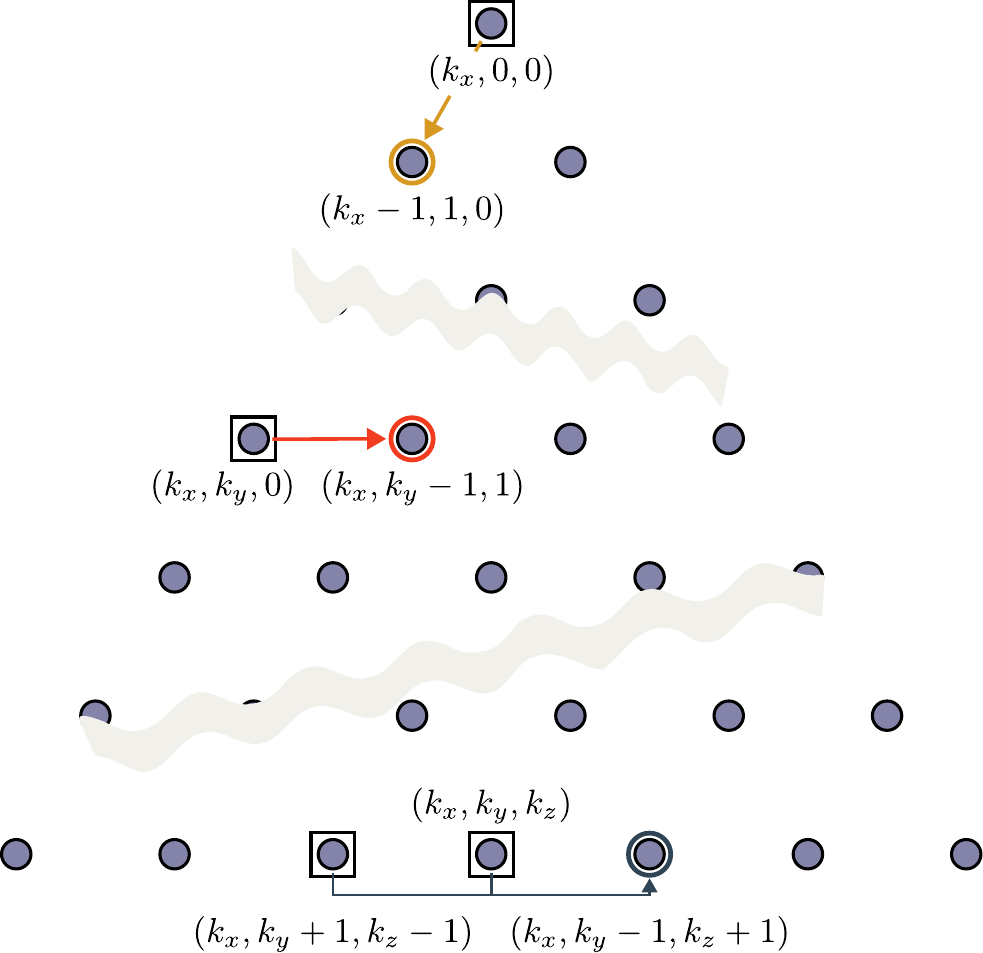}
    \caption{\textbf{Examples movements in a barycentric lattice.} In the figure we schematically show different possible movements in the barycentric lattice as allowed by Lemma~\ref{lem:lattice-hop-one}. At the top, the orange arrow shows that if one has a symmetrized Pauli string on a corner of the level, one can generate an adjacent symmetrized Pauli string. In the middle, the red arrow shows that if one has a  symmetrized Pauli string on the edge of a level, one can generate another symmetrized Pauli string one step inward in the same level of the lattice. Finally, in the bottom we show with black arrows that if one has two adjacent symmetrized Pauli operators, one can generate another one in the same level. 
    \label{fig:proof_3}}
\end{figure}

In the case where $k_z\le 1$, the method is exactly the same. Simply take the commutator with $P_{(1,0,0)}$ twice. The result is a linear combination of $P_{(k_x,k_y,k_z)}$ and $P_{(k_x,k_y-2,k_z+2)}$, where the coefficient of the latter is nonzero. (The only difference from the case above is that $P_{(k_x,k_y+2,k_z-2)}$ no longer exists, since $k_z\le 1$.) Since we are given that $P_{(k_x,k_y,k_z)}\in\mf{g}$, we conclude that $P_{(k_x,k_y-2,k_z+2)}\in\mf{g}$, as desired.
\end{proof}

We now use Lemmas \ref{lem:lattice-hop-one} and \ref{lem:lattice-hop-two} to efficiently hop around the barycentric lattice and construct the operators in $\mf{g}_2$. This idea is captured in the following lemma and in Sup. Fig.~\ref{fig:proof_4}:

\begin{lemma}
\label{lem:fill-out-odd-within-level}
If $P_{(k-1,1,0)}\in\mf{g}_2$, $P_{(0,k-1,1)}\in\mf{g}_2$, and  $P_{(1,0,k-1)}\in\mf{g}_2$, then $P_{(k_x,k_y,k_z)}\in\mf{g}_2$ for any ordered triple $(k_x,k_y,k_z)$ at level $k$ where at least one of  $k_x, k_y, k_z$ is odd. In particular, if $k = k_x + k_y + k_z$ is odd, then any ordered triple at level $k$ will satisfy this, which means that every operator at level $k$ is in $\mf{g}_2$.
\end{lemma}

\begin{proof}
In this proof we will first show that if then $P_{(k-1,1,0)}\in\mf{g}_2$, then $P_{(k_x,k_y,k_z)}\in\mf{g}_2$ for any ordered triple $(k_x,k_y,k_z)$ at level $k$ where $k_y$ is odd. Once this result is established, we can use Lemma~\ref{lem:invariant-under-XYZ-interchange} to obtain the remaining cases. 

\vspace{0.5\baselineskip}

We first prove that $P_{(k-(2j+1),2j+1,0)}\in\mf{g}_2$ for all $0\le j\le\big\lfloor\frac{k-1}{2}\big\rfloor$. To do this, we simply repeatedly apply Lemma \ref{lem:lattice-hop-two}. If $P_{(k-1,1,0)}\in\mf{g}_2$, then $P_{(k-3,3,0)}\in\mf{g}_2$ by Lemma \ref{lem:lattice-hop-two}. Then, for each $2\le j\le\big\lfloor\frac{k-1}{2}\big\rfloor$ in sequence, we can use the fact that $P_{(k-(2j-3),2j-3,0)}\in\mf{g}_2$ and $P_{(k-(2j-1),2j-1,0)}\in\mf{g}_2$ to conclude using Lemma \ref{lem:lattice-hop-two} that $P_{(k-(2j+1),2j+1,0)}\in\mf{g}_2$. Therefore, we can generate $P_{(k-(2j+1),2j+1,0)}$ for all $0\le j\le\big\lfloor\frac{k-1}{2}\big\rfloor$.

\vspace{0.5\baselineskip}

Now, given some odd $k_y$, we know that for any $k_x$ and $k_z$ such that $k_x+k_y+k_z=k$, then $P_{(k_x+k_z,k_y,0)}\in\mf{g}_2$. We can now repeatedly apply Lemma \ref{lem:lattice-hop-one} to repeatedly subtract $1$ from the first coordinate and add $1$ to to the third coordinate to ultimately conclude that $P_{(k_x,k_y,k_z)}\in\mf{g}_2$ where   $k_y$ is odd. As previously mentioned, the rest of the cases follow from Lemma~\ref{lem:invariant-under-XYZ-interchange}.
\end{proof}

So far, we have proved that, if we can generate one particular operator at level $k$, e.g., $P_{(k-1,1,0)}$, then we can generate many other operators at the same level (i.e., all the $P_{(k_x,k_y,k_z)}\in\mf{g}_2$ at level $k$ such that  $k_y$ is odd). In fact, if $k$ is odd, and one has $P_{(k-1,1,0)}$, $P_{(0,k-1,1)}$ and $P_{(1,k-1,0)}$, then one can generate all other operators at level $k$. What remains is to prove that you can actually generate $P_{(k-1,1,0)}$, $P_{(0,k-1,1)}$ and $P_{(1,k-1,0)}$ for each $k$. We take care of this in the following lemma:

\begin{figure}[t]
    \centering
\includegraphics[width=.5\columnwidth]{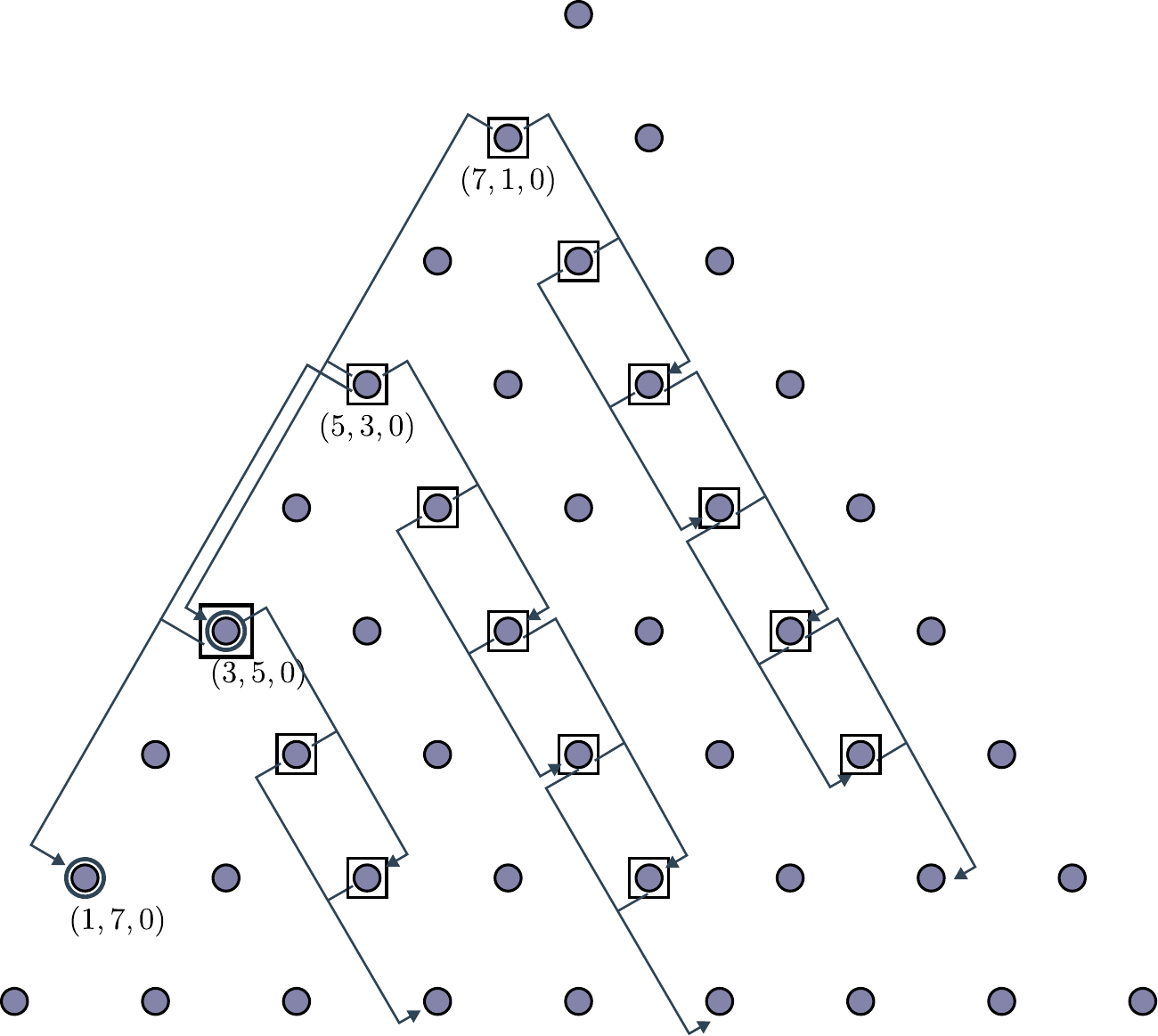}
    \caption{\textbf{Schematic diagram for the proof of Lemma~\ref{lem:fill-out-odd-within-level}}. In this figure we exemplify the key steps used in deriving Lemma~\ref{lem:fill-out-odd-within-level} for the barycentric lattice with $k=8$. We start with $P_{(k-1,1,0)}$, which and we repeatedly apply Lemma~\ref{lem:lattice-hop-two} to generate $P_{(k-3,3,0)}$, $P_{(k-5,5,0)}$, and so on. Then, we use these operators and Lemma~\ref{lem:lattice-hop-one} to generate $P_{(k_x,k_y,k_z)}$ for all $k_x+k_y+k_z=k$ with $k_y$ being odd. By Lemma~\ref{lem:invariant-under-XYZ-interchange}, this suffices to generate    all $P_{(k_x,k_y,k_z)}$ for all $k_x+k_y+k_z=k$ where at least one $k_x$, $k_y$ or $k_z$ is odd.
    \label{fig:proof_4}}
\end{figure}

\begin{lemma}
\label{lem:any-triple-at-least-1-odd}
If $(k_x,k_y,k_z)$ is any ordered triple such that at least one of $k_x, k_y, k_z$ is odd, then $P_{(k_x,k_y,k_z)}\in\mf{g}_2$.
\end{lemma}

\begin{proof}
We prove this by strong induction on the level $k = k_x + k_y + k_z$. We start with the base cases $k=0$ and $k=1$. If $k=0$, then $k_x = k_y = k_z = 0$, so there are no applicable ordered triples. If $k=1$, then the ordered triples are $(k_x,k_y,k_z) = (1,0,0)$ and its permutations. But $P_{(1,0,0)}$ is already a generator of $\mf{g}_2$, so we are done with the base case.

\vspace{0.5\baselineskip}

Now we proceed to the inductive step. Suppose that we have proved that $P_{(k_x,k_y,k_z)}\in\mf{g}_2$ for all ordered triples on level at most $k$ where at least one of $k_x, k_y, k_z$ is odd. Then by the inductive hypothesis, $P_{(k-1,0,1)}$ and $P_{(k-2,1,0)}$ are both in $\mf{g}_2$. Therefore, recalling that $P_{(2,0,0)}\in \mf{g}_2$, we have that the commutator
\begin{equation}
    \left[P_{(k-1,0,1)}, P_{(2,0,0)}\right] \propto (k)P_{(k,1,0)} + (n-k+1)P_{(k-2,1,0)}\,,
\end{equation}
implies that $P_{(k,1,0)}\in\mf{g}_2$. Then Lemmas~\ref{lem:invariant-under-XYZ-interchange} and~\ref{lem:fill-out-odd-within-level} imply that all operators $P_{(k_x,k_y,k_z)}$ on level $k+1$ with at least one of $k_x, k_y, k_z$ odd are contained in $\mf{g}_2$. This completes the induction and hence the proof.
\end{proof}

\begin{figure}[t]
    \centering
\includegraphics[width=1\columnwidth]{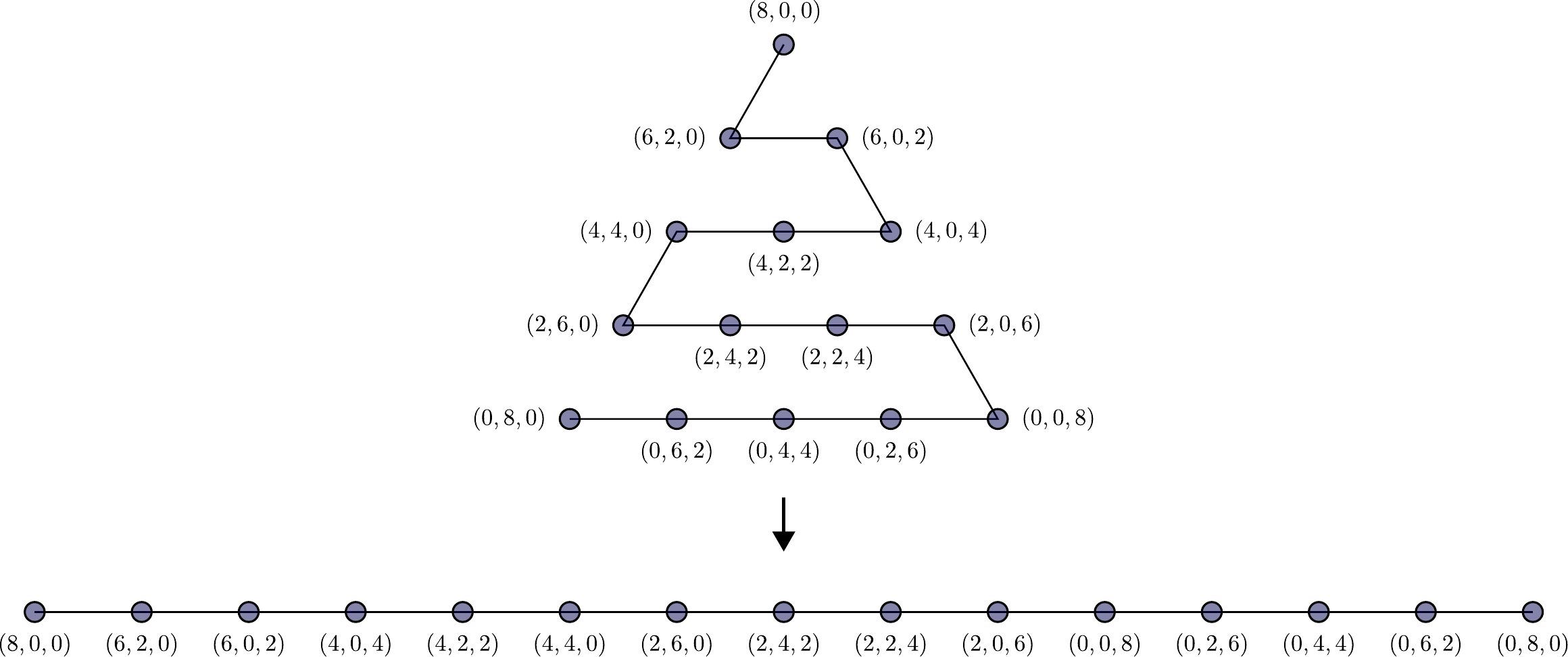}
\caption{\textbf{Ordering symmetrized Pauli strings in the barycentric lattice.} In the figure we schematically show how a ``snaking'' pattern can be created in the barycentric lattice that lines up all operators $P_{(2a,2b,2c)}$ such that $a+b+c=k$ for some fixed value of $k$. Notice that the operators have a distance of $2$ from their neighbours in the line-up (here we have shown $k=4$). The previous means that we can create a two-coordinate operator as in for each pair of neighbours in the line-up. 
    \label{fig:proof_5}}
\end{figure}

But what about the ordered triples $(k_x,k_y,k_z)$ where all three of these values are even? Well, these are exactly the ordered triples that show up in the central elements $C_\mu$, so we can expect these to be restricted by central projections, as articulated in Theorem~2. The one exception is if $(k_x,k_y,k_z)$ is a permutation of $(2,0,0)$, since we are already given $P_{(0,0,2)}$ as a generator. Sure enough, we already generate all of these operators in the proof of Lemma \ref{lem:invariant-under-XYZ-interchange}. Other than those exceptions, we will ultimately be unable to generate $P_{(k_x,k_y,k_z)}$ operators in isolation if $k_x,k_y,k_z$ are all even. However, as we will see, we can still generate all combinations of these operators that satisfy the condition of Theorem~4. This is what we show now, and it is the last step of the proof:

\begin{lemma}
\label{lem:all-even-triples-orthogonal-to-Ck}
For each integer $0\le \mu\le\lfloor\frac{n}{2}\rfloor$ such that $\mu\neq 1$, $\sum_{a+b+c=\mu}c_{(2a,2b,2c)}P_{(2a,2b,2c)}\in\mf{g}_2$ for any set of coefficients $c_{(2a,2b,2c)}$ such that $\sum_{a+b+c=\mu}\frac{c_{(2a,2b,2c)}}{a!b!c!}=0$.
\end{lemma}

\begin{proof}
For $\mu=0$, note that $C_0 = I$, and the singular constraint is $c_{(0,0,0)} = 0$, so we do not need to generate anything here. From now on, assume that $\mu\ge 2$.

\vspace{0.5\baselineskip}

First, we show that, if $(2a,2b,2c)$ and $(2a',2b',2c')$ are two ordered triples on the same level that are separated by a distance of just $2$ on the barycentric lattice, then we can produce a linear combination using just $P_{(2a,2b,2c)}$ and $P_{(2a',2b',2c')}$ that lies in $\mf{g}_2$. We will call these operators \textbf{two-coordinate operators} for convenience. For example, suppose that $a'=a$, $b'=b+1$, and $c'=c-1$. Then we can take the following commutator:
\begin{equation}
    [P_{(2a,2b+1,2c-1)}, P_{(1,0,0)}] \propto (-2c)P_{(2a,2b,2c)} + (2(b+1))X_{2a,2(b+1),2(c-1)}.
\end{equation}
Notice that the two coefficients, $-2c$ and $2(b+1)$, indeed satisfy the required condition:
\begin{equation}
    \frac{-2c}{a!b!c!} + \frac{2(b+1)}{a!(b+1)!(c-1)!} = 0.
\end{equation}
If one  wishes to do repeat the same calculation as above, but for two coordinates separated in the direction that keeps the $y$-coordinate or $z$-coordinate constant, then one needs to take a commutator with $P_{(0,1,0)}$ or $P_{(0,0,1)}$, respectively.

\vspace{0.5\baselineskip}

We will now use these two-coordinate operators to produce an arbitrary linear combination of the form given in the statement of the lemma. In particular we will start by drawing an analogy to a very simple linear algebra exercise:  In $\mathbb{R}^n$, with $e_j$ for $1\le j\le n$ as the standard basis vectors, it is easy to show that the $n-1$ vectors $v_j = e_j - e_{j+1}$ span the space of all vectors whose components add up to $0$. In particular, suppose one wants to produce the vector $(c_1,\cdots,c_n)$ with $c_1 + \cdots + c_n = 0$. Then one just take the appropriate amount of $v_1$ to match the first component, then the appropriate amount of $v_2$ to match the second component, and so on, until one takes the appropriate amount of $v_{n-1}$ to match the $(n-1)^\text{th}$ component. Since $c_1 + \cdots + c_n = 0$, the $n^\text{th}$ component will automatically equal $c_n$. As we can see, following the previous procedure, one  successfully produces the vector $(c_1,\cdots,c_n)$.

\vspace{0.5\baselineskip}

To complete our proof, we will imitate the linear algebra exercise we just presented. Take the $N_\mu = \binom{\mu+2}{2}$ operators $P_{(2a,2b,2c)}$ such that $a+b+c=k$ for a fixed value of $\mu$, and line them up in an order such that each coordinate has distance $2$ from its nearest neighbors in the line. For each $j$ from $1$ to $N_\mu$, label the coordinate at position $j$ in the line as $(2a_j,2b_j,2c_j)$. There are many ways to do this, but one simple way is a ``snaking'' pattern that divides the operators into rows based on decreasing values of $a$. As schematically shown in Sup. Fig.~\ref{fig:proof_5}, one  starts with the first row, which is just $P_{(2\mu,0,0)}$. Then goes to the next row, which will have $P_{(2(\mu-1),2,0)}$ and $P_{(2(\mu-1),0,2)}$ in that order. Then one transitions to the next row and traverse it backward, so that one crosses  $P_{(2(\mu-2),0,4)}$, $P_{(2(\mu-2),2,2)}$, and $P_{(2(\mu-2),4,0)}$, in that order. If one follows this pattern of alternating traversing each successive row forward or backward, one will eventually  reached the end of the lattice (see Sup. Fig.~\ref{fig:proof_5}).

\vspace{0.5\baselineskip}

Now, for each $j$ from $1$ to $N_\mu-1$, we can construct a two-coordinate operator $T_j$ out of $P_{(2a_j,2b_j,2c_j)}$ and $P_{(2a_{j+1},2b_{j+1},2c_{j+1})}$. Now, suppose we wish to generate the operator $M = \sum_{a+b+c=\mu}c_{(2a,2b,2c)}P_{(2a,2b,2c)}$ for any set of coefficients $c_{(2a,2b,2c)}$ such that $\sum_{a+b+c=\mu}\frac{c_{(2a,2b,2c)}}{a!b!c!}=0$. Then start with the right amount of $T_1$ so that the coefficient of $P_{(2a_1,2b_1,2c_1)}$ matches that in $M$. Then add the right amount of $T_2$ so that the coefficient of $P_{(2a_2,2b_2,2c_2)}$ matches that in $M$. Keep doing this, until you finally add the right amount of $T_{N_\mu-1}$ so that the coefficient of $P_{(2a_{N_k-1},2b_{N_\mu-1},2c_{N_\mu-1})}$ matches that in $M$. Due to the central projections condition, the coefficient of $P_{(2a_{N_\mu},2b_{N_\mu},2c_{N_\mu})}$ will now automatically match that in $M$. Therefore, we have constructed our desired operator, so we are done.
\end{proof}

The proof is now complete, so let us summarize what we have done. As shown in Theorem~2 of the main text, we first used the central projections condition to show that being orthogonal to every $C_\mu$ except $\mu=1$ is a necessary condition for an operator to be in $\mf{g}_2$. After that, we showed that $\mf{g}_2$ indeed contains all operators that satisfy that condition. As shown in Lemma \ref{lem:any-triple-at-least-1-odd}, we successfully generated $P_{(k_x,k_y,k_z)}$ where at least one of $k_x,k_y,k_z$ is odd. The form of the $C_\mu$ operators makes it clear that all such $P_{(k_x,k_y,k_z)}$ are already orthogonal to all of the $C_\mu$ operators. Finally, as shown in Lemma \ref{lem:all-even-triples-orthogonal-to-Ck}, we successfully generated all linear combinations of the $P_{(2a,2b,2c)}$ operators on level $2\mu$ that are orthogonal to $C_\mu$. Note that we could ignore $\mu=1$ because we already showed in Lemma \ref{lem:invariant-under-XYZ-interchange} that $\mf{g}_2$ contains all $6$ level-$2$ operators.

\vspace{0.5\baselineskip}

The only operators in $\mf{u}^{S_n}(d)$ that are "missing" from $\mf{g}_2$ are the $C_\mu$ operators for $0\le \mu\le\lfloor n/2\rfloor$ except $\mu=1$. This includes $C_0 = I$, but it also includes the operators $C_\mu$ for $2\le \mu\le\lfloor n/2\rfloor$. As a result, we obtain $\text{dim}(\LC) = \binom{n+3}{3} - \lfloor n/2\rfloor$.

\vspace{0.5\baselineskip}

As a final observation for this section, note that the requirement of orthogonality with $C_0 = I$ is what prevents any operator with a nonzero $P_{(0,0,0)} = I$ component from being in $\mf{g}_2$. Of course, the identity will be in the center of the commutant of any set of generators, so the central projections condition of Ref. \cite{zimboras2015symmetry} will always exclude operators with nonzero trace from being in the algebra if the generators are all traceless. Of course, the fact that the identity cannot be in such an algebra is by itself a pretty trivial observation, since it is such a common fact that the commutator of two finite-dimensional operators is traceless. But this work already serves to demonstrate that the central projections condition can be used to prove a highly non-obvious result. It is interesting to note that the central projection generalizes such a seemingly mundane statement as the fact that the commutator of two finite dimensional operators is traceless, and yet it is powerful enough to crop up in unexpected ways.

\section{Proof of Theorem~5}\label{sec:theo:k-body-DLA}

Let us now prove Theorem~5, which we restate for convenience.

\begin{theorem}
\label{thm:k-body-DLA-ap}
Consider the set $\GC_k$ of $S_n$-equivariant generators in Eq.~(36). Then, the associated DLA $\mf{g}_k$ is
\begin{align}
\mf{g}_k&=Q\left(\bigoplus_{\lambda=1}^L\mf{su}(m_\lambda)\oplus\underbrace{\mf{u}(1)\oplus\cdots\oplus\mf{u}(1)}_{\lfloor k/2\rfloor}\right)\\
&=\mf{su}^{S_n}_{\rm{cless}}(d) \boxplus Q\left( \underbrace{\mf{u}(1)\oplus\cdots\oplus\mf{u}(1)}_{\lfloor k/2\rfloor}  \right) \,,
\end{align}
where $Q\left( \mf{u}(1)\oplus\cdots\oplus\mf{u}(1)  \right) = \spn_{\mbb{R}}\{iC_1,\cdots,iC_{\lfloor k/2\rfloor}\}$ is a $\lfloor k/2\rfloor$-dimensional subalgebra of $\mf{z}(\mf{u}^{S_n}(d))$.
\end{theorem}

Furthermore, $\mf{g}_k$ contains all linear combinations of the form
\begin{equation}
    i\sum_{0\le k_x+k_y+k_z\le n}c_{(k_x,k_y,k_z)}P_{(k_x,k_y,k_z)},
\end{equation}
for a collection of real coefficients $c_{(k_x,k_y,k_z)}$ that satisfy
\begin{equation}
    \sum_{a+b+c=\mu}\frac{c_{(2a,2b,2c)}}{a!b!c!} = 0\,,
\end{equation}
for $\mu=0$ and each integer $\lfloor k/2\rfloor +1\le \mu\le\lfloor n/2\rfloor$.

\begin{proof}
Since $\GC_k$ contains the generators $P_{(1,0,0)}$, $P_{(0,1,0)}$ and $P_{(0,0,2)}$, $\mf{g}_k$ contains $\mf{g}_2$. Therefore, by Theorem 4, $\mf{g}_k$ contains all linear combinations of the form
\begin{equation}
    i\sum_{0\le k_x+k_y+k_z\le n}c_{(k_x,k_y,k_z)}P_{(k_x,k_y,k_z)},
\end{equation}
for a collection of real coefficients $c_{(k_x,k_y,k_z)}$ that satisfy
\begin{equation}
    \sum_{a+b+c=\mu}\frac{c_{(2a,2b,2c)}}{a!b!c!} = 0\,,
\end{equation}
for each integer $0\le \mu\le\lfloor n/2\rfloor$ except $\mu=1$.

\vspace{0.5\baselineskip}

Let us now take advantage of our higher-body generators. For each $1\le \mu\le k$, since $\mf{g}_k$ contains the generator $P_{(0,0,\mu)}$, along with the one-body generators $P_{(0,0,1)}$ and $P_{(0,1,0)}$, $\mf{g}_k$ contains the full span of symmetrized Pauli strings at level $\mu$, since we can repeatedly apply Lemma \ref{lem:lattice-hop-one} to hop around one step at a time on level $\mu$ of the barycentric lattice until we have reached every symmetrized Pauli string at level $\mu$. This means that, among the conditions on the coefficients that come from the central projections condition, i.e.,
\begin{equation}
    \sum_{a+b+c=\mu}\frac{c_{(2a,2b,2c)}}{a!b!c!} = 0\,,
\end{equation}
for each integer $0\le \mu\le\lfloor n/2\rfloor$ except $\mu=1$, we are now able to additionally disregard the condition for $1\le \mu\le\lfloor k/2\rfloor$. In particular, since the central element $C_{\mu}$ lives entirely on level $2\mu$, we can now generate $C_{\mu}$ for every $1\le \mu\le\lfloor k/2\rfloor$. Since these are central elements, and there are $\lfloor k/2\rfloor$ of them, we know that these elements span some representation of a direct sum of $\lfloor k/2\rfloor$ algebras $\mf{u}(1)$.

\vspace{0.5\baselineskip}

In contrast, for $\mu=0$ and $\lfloor k/2\rfloor + 1\le \mu\le\lfloor n/2\rfloor$, the central projections condition still implies that $\Tr[MC_{\mu}] = 0$ for every $M\in\mf{g}_k$, because none of the generators in $\GC_k$ share any Pauli strings in common with $C_{\mu}$ for $\mu=0$ and $\lfloor k/2\rfloor + 1\le \mu\le\lfloor n/2\rfloor$. It follows that, for every $M\in\mf{g}_k$, its coefficients $c_{(k_x,k_y,k_z)}$ satisfy
\begin{equation}
    \sum_{a+b+c=\mu}\frac{c_{(2a,2b,2c)}}{a!b!c!} = 0\,,
\end{equation}
for $\mu=0$ and each integer $\lfloor k/2\rfloor + 1\le \mu\le\lfloor n/2\rfloor$.

\vspace{0.5\baselineskip}

By Theorem~6, we know that
\begin{align}
\mf{g}_2
&=\mf{su}^{S_n}_{\rm{cless}}(d)\boxplus Q\left(\mf{u}(1)\right)\,,
\end{align}
where $Q$ is some representation of $\mf{u}(1)$.

However, the only difference between $\mf{g}_k$ and $\mf{g}_2$ is the inclusion of $\lfloor k/2\rfloor - 1$ new central elements $C_{2\mu}$ for $2\le \mu\le\lfloor k/2\rfloor$. Therefore, the only possibility for the decomposition of $\mf{g}_k$ is
\begin{align}
\mf{g}_k
&=\mf{su}^{S_n}_{\rm{cless}}(d)\boxplus Q_k\left(\underbrace{\mf{u}(1)\oplus\cdots\oplus\mf{u}(1)}_{\lfloor k/2\rfloor}\right)\,,
\end{align}
where $Q_k$ is some representation  of  a direct sum of $\lfloor k/2\rfloor$ algebras $\mf{u}(1)$.
\end{proof}

\section{Proof of Theorem~6}\label{sec:theo:central-projection-k-body}

Let us now prove Theorem~6, which we restate for convenience.

\setcounter{theorem}{5}
\begin{theorem}
\label{theo:central-projection-k-body-sm}
Consider the set $\GC_k$ of $S_n$-equivariant generators in Eq.~(36). Then, one has that
\begin{equation}
    \Tr[H C_\mu]=0\,,
\end{equation}
for all $H\in\mathcal{G}_k$ and for all $C_\mu$ with $\mu=0$ and  $\lfloor \frac{k}{2}\rfloor< \mu\le\lfloor n/2\rfloor$.
\end{theorem}

\begin{proof}
    The proof of this theorem follows that of Theorem~3, where it suffices to note that none of the generators in $\GC_k$ share any Pauli strings in common with any $C_\mu$, with the exception of $P_{(0,0,\mu)}$.
\end{proof}

\section{Error in the proof of Ref.~\cite{albertini2018controllability}}\label{sec:error}

In our main results we show that the DLA ensuing from $\mf{g}_2$ is semi-universal  and subspace controllable but fails to be universal. This result clashes with that of  Ref.~\cite{albertini2018controllability} where the authors claim that the DLA is universal, i.e., that $\mf{g}_2=\mf{su}^{S_n}(d)$. 

In this section we will follow we will follow the notation used in Ref.~\cite{albertini2018controllability} for ease of comparison of results. Hence, we will now redefine $Y\rightarrow - Y$ and $P_{(k_x,k_y,k_z)}\rightarrow i P_{(k_x,k_y,k_z)}$.

 We can narrow down the bug in the proof of Ref.~\cite{albertini2018controllability} as spanning from the following commutation relations 
\begin{align}
\left[P_{(\overline{k}-1,1,0)}, P_{(0,0,1)}\right] &= (-2)\textcolor{red}{(2)}P_{(\overline{k}-2,2,0)} + (2)\textcolor{red}{(\overline{k})}P_{(\overline{k},0,0)}\,, \label{eq:23}\\
    \left[P_{(\overline{k}-1,0,1)}, P_{(0,1,0)}\right] &= (2)\textcolor{red}{(2)}P_{(\overline{k}-2,0,2)} + (-2)\textcolor{red}{(\overline{k})}P_{(\overline{k},0,0)}\,,\label{eq:24}\\
\left[P_{(\overline{k}-2,1,0)}, P_{(1,0,1)}\right] =& (-2)\textcolor{red}{(2(n-\overline{k}+2))}P_{(\textcolor{red}{\overline{k}-4},2,0)}+ (2)\textcolor{red}{((\overline{k}-2)(n-\overline{k}+2))}P_{(\overline{k}-2,0,0)} \nonumber\\
    &+ (-2)\textcolor{red}{(2(\overline{k}-2))}P_{(\overline{k}-2,2,0)} + (-2)\textcolor{red}{(2)}P_{(\overline{k}-2,0,2)} + (2)\textcolor{red}{(\overline{k}(\overline{k}-1))}P_{(\overline{k},0,0)}.\label{eq:26}
\end{align}
Here we have added in red the coefficients and indexes that were incorrectly derived in Ref.~\cite{albertini2018controllability}. Let us explain how we derived the coefficients for the right-hand side of Eq.~\eqref{eq:23}. First, each Pauli string in $P_{(\overline{k}-2,2,0)}$ can be generated in $2$ ways from the commutator on the left, since there are $2$ ways to choose which Pauli $Y$ got introduced. After that, each Pauli string in $P_{(\overline{k},0,0)}$ can be generated in $\overline{k}$ ways from the commutator on the left, since there are $\overline{k}$ ways to choose which one of the $\overline{k}$ Pauli $X$ symbols came from the $[Y,Z]$ commutator. The coefficients in Eqs.~\eqref{eq:24} and~\eqref{eq:26} are derived with similar combinatorial arguments.

Let $A$ and $B$ denote the right-hand-sides of Eqs.~\eqref{eq:23} and~\eqref{eq:24}, respectively:
\begin{align}
    A &= (-2)\textcolor{red}{(2)}P_{(\overline{k}-2,2,0)} + (2)\textcolor{red}{(\overline{k})}P_{(\overline{k},0,0)}\label{eq:23-A} \\
    B &= (2)\textcolor{red}{(2)}P_{(\overline{k}-2,0,2)} + (-2)\textcolor{red}{(\overline{k})}P_{(\overline{k},0,0)}\,.\label{eq:24-B}
\end{align}

Next, the authors of Ref.~\cite{albertini2018controllability} claim that since $P_{(\overline{k}-4,2,0)}$ 
 and $P_{(\overline{k}-2,0,0)}$ are in the DLA, then one can  generate an operator, which we call $C$, and which we define as
\begin{equation}
    C = (-2)\textcolor{red}{(2(\overline{k}-2))}P_{(\overline{k}-2,2,0)} + (-2)\textcolor{red}{(2)}P_{(\overline{k}-2,0,2)} + (2)\textcolor{red}{(\overline{k}(\overline{k}-1))}P_{(\overline{k},0,0)}.
\end{equation}

Without the combinatorial corrections in red, the operators $A$, $B$ and $C$  would be linearly independent, and this is precisely the mistake in~\cite{albertini2018controllability} (see Eqs.~(23), (24) and (26) therein). However, with the combinatorial corrections, these three quantities satisfy the relation
\begin{equation}
    C = (\overline{k}-2)A + (-1)B.
\end{equation}
Hence these three quantities are not linearly independent. This small, albeit important difference ruins the ability to complete the proof in~\cite{albertini2018controllability}.

\end{document}